\documentclass[11pt,a4paper]{article}

\usepackage[utf8]{inputenc}
\usepackage[T1]{fontenc}
\usepackage{amsmath,amssymb,amsthm}
\usepackage{mathtools}
\usepackage{graphicx}
\usepackage[margin=1in]{geometry}
\usepackage{hyperref}
\usepackage{cleveref}
\usepackage[dvipsnames]{xcolor}
\usepackage{enumitem}
\usepackage{booktabs}
\usepackage{algorithm}
\usepackage{algpseudocode}
\usepackage{tikz}
\usetikzlibrary{braids,knots,decorations.markings,automata,positioning,arrows.meta}

\newtheorem{theorem}{Theorem}[section]
\newtheorem{lemma}[theorem]{Lemma}
\newtheorem{proposition}[theorem]{Proposition}
\newtheorem{corollary}[theorem]{Corollary}

\theoremstyle{definition}
\newtheorem{definition}[theorem]{Definition}

\theoremstyle{remark}

\newcommand{\Bn}{B_n}

\newcommand{\NHI}{\mathrm{NHI}}
\newcommand{\DWS}{\operatorname{DWS}}
\newcommand{\LE}{\operatorname{LE}}
\newcommand{\SR}{\operatorname{SR}}

\title{Detecting Privilege Escalation with Temporal Braid Groups}
\author{Christophe Parisel\thanks{Email: ch.parisel@gmail.com}}
\date{\today}

\begin{document}
\maketitle

\begin{abstract}
Within the Strongly Connected Components (SCCs) formed during the temporal
evolution of a Cloud permission graph,
we use the Burau Lyapunov exponent $\LE$ as an algebraic
probe to locate the boundary between two risk regimes and to calibrate cheaper detection
instruments.  We prove that no Abelian statistic (edge counts, net
privilege flow, gate-firing rates) can determine $\LE$; empirically,
over 49,972 deployments drawn from a 1000-SCC corpus, temporal
braid $\LE$ retains partial correlation $r = 0.175$ with deployment
structure after controlling for firing rate ($p < 10^{-3}$), while the
Abelian counting proxy collapses to $r = 0.001$ ($p = 0.98$): zero
residual signal. 

The non-commutation advantage is small, but actionable:
we show how to leverage it to discriminate the two outstanding risk regimes, that we call dispersed
and focused, for automating classification and governing remediation of risky Cloud permission flows.

\end{abstract}

%% ============================================================
\section{Introduction}
\label{sec:intro}

The security posture of a cloud identity is not determined by a
single permission snapshot but by its full temporal trajectory.
Two non-human identities (NHIs) may hold identical permissions
today yet diverge tomorrow: one may ratchet irreversibly toward
broader access while the other oscillates within a bounded envelope.

In a companion paper~\cite{primorial2026}, we introduced the primorial
invariant, which encodes the cycle structure of a Strongly Connected
Component (SCC)~\cite{tarjan1972}  as a vector of rational numbers and provides a
first-pass topological screen: it discriminates \emph{oscillators}
(SCCs whose directed edges cannot sustain monotone privilege escalation)
from \emph{ratchets} (SCCs with at least one directed WAR path capable
of irreversible upward drift).  Oscillators are topologically safe and
require no further dynamic analysis; ratchets carry escalation potential
and must be assessed dynamically.

Within ratchets, a second distinction governs remediation.

A \textbf{focused} ratchet concentrates escalation through few paths;
non-commutative cancellations moderate spectral growth and WAR
reassignment is curative.  A \textbf{dispersed} ratchet has a hub-rich
topology with multiple independent escalation paths; topology surgery
(adding or reversing edges) is required.

The central result is that the Burau Lyapunov exponent $\LE$, while small 
when taken over the whole domain, is the only available separator of these 
two regimes at their shared boundary: 
it is provably beyond
the reach of any Abelian statistic (\Cref{prop:abelian-blindness}).
Burau LE serves as both \emph{boundary separator} 
and \emph{calibration oracle} (validating cheaper
instruments).  The Abelian gate-firing rate provides a fast first-pass.

\newpage
\paragraph{Contributions.}
\begin{enumerate}
  \item \textbf{Impossibility result.}  No Abelian statistic (edge
        counts, net privilege flow, or gate-firing rate) can determine
        $\LE$ (\Cref{prop:abelian-blindness}).  A sharp structural
        instance: any directed cycle in a ratchet SCC contributes
        exactly zero to $\DWS$ for \emph{every} WAR assignment, by a
        telescoping-sum identity (\Cref{cor:telescoping}), so
        directed-cycle topologies have a $\DWS$-invisible escalation
        channel that the braid walk detects unconditionally.  Temporal
        braid $\LE$ retains partial correlation $r = 0.175$ with
        deployment structure after controlling for firing rate
        ($p < 10^{-3}$); the Abelian proxy collapses to $r = 0.001$
        ($p = 0.98$).
  \item \textbf{Two-regime classification with Burau as calibration oracle.}
        $\LE$ defines a focused/dispersed boundary that determines
        which intervention (WAR reassignment vs.\ topology surgery) is
        required.  Across 49,972 (SCC, WAR) pairs from a 1000-SCC corpus,
        5.7\% of pairs disagree between Burau and the Abelian rate,
        with topology-dependent FD/DF structure that the Abelian rate
        cannot predict.
  \item \textbf{Gate condition and braid construction.}  NHI walkers
        co-occurring on ascending directed edges inject braid generator
        $\sigma_i^2\sigma_{i+1}^{-1}$, coupling graph topology to
        privilege-flow alignment (\Cref{sec:construction}).
  \item \textbf{Firing-rate lower bound for hub SCCs.}  For $k$-hub SCCs,
        the firing rate exceeds a topology-derived lower bound
        (\Cref{thm:firing-rate}); under a spoke-coverage condition the
        bound places dense hubs unconditionally in the dispersed regime.
\end{enumerate}

%% ============================================================
\section{Background}
\label{sec:background}

\subsection{Strongly connected components}

A \emph{directed graph} $G = (V,E)$ consists of vertices and
directed edges. A \emph{strongly connected component} (SCC) is a
maximal subgraph in which every vertex is reachable from every
other vertex via directed paths. SCCs are the ``inescapable
neighbourhoods'' of a directed graph: once inside, every vertex can
reach every other.

In graph theory, determining whether a directed graph is strongly
connected is a classical problem (Tarjan's algorithm runs in linear
time). Every directed graph decomposes uniquely into SCCs connected
by a directed acyclic graph (DAG) of inter-component edges.

A \emph{cycle} in a directed graph is a closed walk
$v_0 \to v_1 \to \cdots \to v_k = v_0$. Every SCC with more than
one vertex contains at least one cycle; the cycle structure encodes
recurrence patterns. When vertices carry scalar labels (in our
setting, privilege weights) the label changes along a cycle
determine whether the cycle amplifies, dampens, or preserves the
scalar quantity.

\subsection{Random walks on graphs and Markov chains}
\label{sec:markov}

A \emph{random walk} on a directed graph $G = (V,E)$ is a discrete
Markov chain: from the current vertex $v$, the walker follows one
outgoing edge chosen uniformly at random, reaching neighbour $u$
with probability $P_{vu} = 1/\deg^+(v)$ if $(v,u)\in E$, and
$P_{vu}=0$ otherwise.  Here $\deg^+(v)$ is the out-degree of $v$
and $P \in [0,1]^{|V|\times|V|}$ is the \emph{transition matrix}.

A Markov chain is \emph{irreducible} if every state is reachable
from every other state in a finite number of steps.  An SCC is
precisely the maximal subgraph on which the restricted walk is
irreducible: inside an SCC every vertex can reach every other
(by strong connectivity), while no edge leaves the component.
The random walk confined to an SCC is therefore irreducible by
construction.

By the Perron--Frobenius theorem~\cite{perron}, an irreducible aperiodic Markov
chain has a unique \emph{stationary distribution} $\pi$ satisfying
$\pi P = \pi$ with $\sum_{v}\pi(v) = 1$.  The stationary mass
$\pi(v)$ is the long-run fraction of time the walker spends at $v$.
For a uniform random walk on a strongly connected graph, $\pi(v)$
is proportional to the in-degree of $v$: high-fanin vertices
accumulate the most visit mass.

In our framework, the NHI walkers of \Cref{sec:construction} perform independent random
walks on a fixed ratchet SCC.  Irreducibility ensures that the
long-run walk statistics are well-defined and independent of
starting position.  In particular, the Lyapunov exponent
(\Cref{def:le}) can be consistently estimated from a single
trajectory of sufficient length rather than requiring ensemble
averages.  Stationary mass determines which vertices (hence
which edges) are visited most frequently, shaping how often
the gate condition fires and which braid generators enter the
accumulated product.

\subsection{Braid groups}

The \emph{braid group}~\cite{artin1947}  $B_n$ on $n$ strands is the group generated
by $\sigma_1, \ldots, \sigma_{n-1}$ subject to:
\begin{align}
  \sigma_i \sigma_j &= \sigma_j \sigma_i
    &&\text{if } |i - j| \geq 2,
    \label{eq:braid-commute} \\
  \sigma_i \sigma_{i+1} \sigma_i
    &= \sigma_{i+1} \sigma_i \sigma_{i+1}
    &&\text{(Yang--Baxter relation).}
    \label{eq:yang-baxter}
\end{align}
Geometrically, $\sigma_i$ swaps strands $i$ and $i+1$ with strand
$i$ passing over strand $i+1$, and $\sigma_i^{-1}$ is the
reverse crossing. A \emph{braid word} is a finite product of
generators and their inverses.

\subsection{Non-Commutation}

A group is \emph{Abelian} (commutative) if $ab = ba$ for all
elements. The braid group $B_2$ is generated by a single element
$\sigma_1$ and is therefore Abelian. For $n \geq 3$,
$B_n$ is \emph{non-Abelian}: $\sigma_1 \sigma_2 \neq
\sigma_2 \sigma_1$. This non-commutativity is essential because
it allows products of braid generators to exhibit
\emph{exponential growth} in matrix representations, the key
mechanism underlying our risk metric.

In an Abelian group, matrix representations have eigenvalues on
the unit circle and spectral radius 1: no matter how many
generators you multiply, the product cannot grow. Non-Abelian
words can break this barrier, producing spectral radii strictly
greater than 1 and hence exponential growth under iteration.

\subsection{The Burau representation}

The \emph{Burau representation} is a homomorphism
$\rho_t : B_n \to \mathrm{GL}_n(\mathbb{Z}[t, t^{-1}])$
that sends each braid generator $\sigma_i$ to an $n \times n$
matrix with entries depending on a formal parameter $t$. At the
specialisation $t = -1$, all entries become integers, enabling
exact computation. The spectral radius of the resulting integer
matrix product measures how fast the braid ``stretches'' under
iteration. This is the signal we use to quantify risk.
\Cref{sec:burau} gives the full definition.

In this framework, the full Burau computation ($N \times N$ matrices,
multi-NHI walks) serves as \emph{validation infrastructure}: the
non-Abelian ground truth that proves temporal braid LE captures
genuine non-Abelian structure beyond the Abelian rate.

\newpage

%% ============================================================
\section{Construction: Permission Braids with Signal Injection}
\label{sec:construction}
\subsection{Permission transition graphs and the WAR permission scalar}

Let $G = (V, E)$ be a directed graph where each vertex $v \in V$
represents a permission state characterised by a scalar, the 
\emph{WAR norm}\cite{parisel2025scoring} $W(v) \in \mathbb{Z}_{\geq 0}$,
and each edge $(u, v) \in E$ represents an observed permission
transition. Edges are either \emph{directed} ($u \to v$, irreversible:
$v \to u \notin E$) or \emph{bidirectional} ($u \leftrightarrow v$,
reversible: both $u \to v$ and $v \to u \in E$).

\begin{definition}[Directed WAR flow]
\label{def:war-flow}
For an edge $(u \to v) \in E$, the \emph{WAR flow} is
$\delta(u \to v) = W(v) - W(u)$. A directed edge is
\emph{ascending} if $\delta > 0$, \emph{flat} if $\delta = 0$,
and \emph{descending} if $\delta < 0$.
\end{definition}

\subsection{SCC classification}

We classify SCCs along two axes:
\begin{itemize}
  \item \textbf{Topology axis}: does the SCC contain at least one
        directed edge?
  \item \textbf{Privilege axis}: does the current WAR assignment align
        with the directed edges (ascending flow)?
\end{itemize}

SCCs with no directed edges are \emph{periodic}: WAR trajectories
are mean-reverting, and the braid word produced by walkers on such
SCCs is Abelian (all generators commute). These are structurally 
benign and are fully characterised
by the primorial period-2 oscillator class~\cite{primorial2026}.

SCCs with at least one directed edge are \emph{ratchets}. The present
paper focuses exclusively on ratchets, since these are the SCCs for
which the deployed WAR assignment determines potential severity.

\subsection{NHI walkers and strand ordering}

Fix a ratchet SCC $C$ with $n$ walkers $\NHI_1, \ldots, \NHI_n$
performing simultaneous random walks over $V(C)$.
The start positions of each NHI is random, since SCCs do not have a 
defined entry point (unlike DAGs).
At each epoch $t$, the walkers are ordered by their current WAR
values: strand $i$ is assigned to the walker with the $i$-th
smallest WAR value (ties broken by a fixed NHI identifier order).

The simultaneous walker construction is a mathematical device used to
generate braid words from a static SCC topology under a given WAR
assignment. It does not model real-world concurrency of NHIs, nor does
it assume that two identities escalate privileges at the same time in
production.

The walkers serve as a sampling mechanism: their concurrent motion
induces non-commuting generators capturing
structural interactions between directed edges under WAR ordering.

\subsection{Gate condition and braid word injection}
\label{sec:gate}

When two adjacent walkers both step along directed edges in the same epoch, with both edges
ascending in privilege, the ratchet has fired: two independent ascending directed edges coexist locally in the WAR-ordered topology. 
The injected braid word encodes this structural configuration algebraically.

\begin{definition}[Gate condition]
\label{def:gate}
At epoch transition $t \to t+1$, the \emph{gate fires} for walkers
at ranks $i$ and $i+1$ if and only if:
\begin{enumerate}
  \item Both walkers traverse \emph{directed} edges
        (neither edge is bidirectional), and
  \item Both edges have positive net WAR flow:
        $\delta(\text{edge}_i) > 0$ and
        $\delta(\text{edge}_{i+1}) > 0$.
\end{enumerate}
\end{definition}

The gate fires only when both axes are satisfied at once: the edges
must be truly irreversible (directed, not bidirectional), and the
privilege flow must run uphill. Either condition alone is
insufficient: a directed edge with descending WAR represents a
privilege drop, not an escalation.

The gate condition should be understood as an algebraic trigger
that enriches the braid word when two ascending directed edges coexist
locally in the topology. It is not a claim about synchronized privilege
escalation events in operational systems.

\begin{definition}[Permission braid with word injection]
\label{def:perm-braid}
The \emph{permission braid} $\beta_C \in \Bn$ is an algebraic probe of SCC 
potential severity.  It is constructed as
follows. At each epoch $t \to t+1$, for each adjacent pair of walkers
at ranks $i$ and $i+1$:
\begin{enumerate}
  \item If the gate fires (\Cref{def:gate}): append
        $\sigma_i^2 \sigma_{i+1}^{-1}$ to the braid word.
  \item Otherwise: append nothing (the pair contributes no generators
        at this step).
\end{enumerate}
Rank swaps that do not satisfy the gate condition are deliberately
excluded. A lone generator $\sigma_i^{\pm 1}$ has Burau eigenvalues
on the unit circle ($\mathrm{SR} = 1$ at $t = -1$), so it contributes
no spectral growth. Including such generators would increase the
matrix product length without adding signal.
\end{definition}

\subsection{The injection word}
\label{sec:word-choice}

The injection word must have spectral radius $> 1$ under the Burau
representation at $t = -1$, so that repeated gate firings drive
exponential growth of the accumulated matrix product.

Any word with SR $> 1$ must involve at least two \emph{adjacent}
generator indices: single-index words stay Abelian (SR $= 1$), and
non-adjacent generators commute (SR $= 1$). The shortest candidate
is $\sigma_i \sigma_{i+1}^{-1}$ (length~2,
SR $= \phi^2 \approx 2.618$), but many non-Abelian words of
length~3 also qualify. Among them, we select
$\sigma_i^2 \sigma_{i+1}^{-1}$ for its higher spectral radius:
SR $= 2 + \sqrt{3} \approx 3.732$, a $1.43\times$ lift over the
length-2 word (\Cref{prop:sr-injection}).

\paragraph{Adjacent-position cancellation.}
All words of the form $\sigma_i^a \sigma_{i+1}^b$ with opposite-sign 
exponents share a structural limitation: when two such words fire at adjacent
positions $i$ and $i+1$ within the same epoch, the inner generators
cancel algebraically at least partially, reducing the spectral radius. For our
word:
$\sigma_i^2 \sigma_{i+1}^{-1} \cdot \sigma_{i+1}^2
 \sigma_{i+2}^{-1} = \sigma_i^2 \sigma_{i+1} \sigma_{i+2}^{-1}$,
with SR reduced to~$1$ (the Burau matrix of this word has finite order~$6$
in $\mathrm{GL}(3,\mathbb{Z})$, so all its eigenvalues are roots of unity).
The same holds for the length-2 alternative:
$\sigma_i \sigma_{i+1}^{-1} \cdot \sigma_{i+1} \sigma_{i+2}^{-1}
 = \sigma_i \sigma_{i+2}^{-1}$, also yielding SR $= 1$ (non-adjacent generators commute and are individually
unipotent at $t=-1$, so their product is unipotent).

\medskip

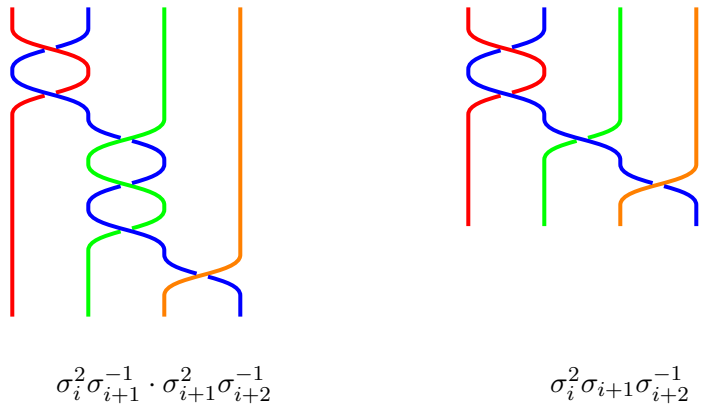
\begin{figure}[h]
\centering
\begin{tikzpicture}[scale=1]
% --- Left braid word ---
\pic[
    line width=1.5pt,
    braid/.cd,
    width=1cm,
    crossing height=.6cm,
    strand 1/.style={red},
    strand 2/.style={blue},
    strand 3/.style={green},
    strand 4/.style={orange},
] at (0,0) {braid={s_1 s_1 s_2^{-1} s_2 s_2 s_3^{-1}}};

\node at (2,-5) {$\sigma_i^2 \sigma_{i+1}^{-1} \cdot \sigma_{i+1}^2 \sigma_{i+2}^{-1}$};

% --- Right braid word ---
\pic[
    line width=1.5pt,
    braid/.cd,
    width=1cm,
    crossing height=.6cm,
    strand 1/.style={red},
    strand 2/.style={blue},
    strand 3/.style={green},
    strand 4/.style={orange},
] at (6,0) {braid={s_1 s_1 s_2 s_3^{-1}}};

\node at (8,-5) {$\sigma_i^2 \sigma_{i+1} \sigma_{i+2}^{-1}$};

\end{tikzpicture}
\caption{Adjacent-position cancellation:  
$\sigma_i^2 \sigma_{i+1}^{-1} \cdot \sigma_{i+1}^2 \sigma_{i+2}^{-1}$  
collapses to a Burau-flat braid with spectral radius $1$.}  
\end{figure}

The adjacent-pair guard (\Cref{sec:gate}) prevents this within-epoch
cancellation by skipping position $i+1$ whenever position $i$ fires.
Cross-epoch interactions at adjacent positions are not cancelled
algebraically: they compose non-trivially, producing the spectral
cancellation that the shuffling experiment detects
(\Cref{sec:shuffling}).

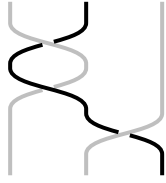
\begin{figure}[h]
\centering
\begin{tikzpicture}[scale=1]
\pic[
    line width=1.5pt,
    braid/.cd,
    width=1cm,
    crossing height=.6cm,
    strand 1/.style={lightgray},
    strand 2/.style={black},
    strand 3/.style={lightgray},
] at (6,0) {braid={s_1 s_1 s_2^{-1}}};

\end{tikzpicture}
\caption{The injection word $\sigma_1^2\sigma_2^{-1}$
  on three strands. Two positive crossings of strands 1--2 followed
  by one negative crossing of strands 2--3 produce a genuinely
  non-Abelian braid with SR $= 2+\sqrt{3}$.}
\end{figure}

Each time the gate fires,
the injection word's Burau matrix is multiplied into the running
product. The dominant eigenvalue governs the per-injection growth
factor of that product: if it exceeds 1, the spectral radius of
the accumulated matrix grows exponentially in the number of
gate firings, producing a positive Lyapunov exponent. A dominant eigenvalue of exactly 1 would mean
no growth at all (the braid is ``flat''). The further above 1, the
faster the braid stretches under iteration, and the stronger the
risk signal.

\begin{proposition}[Spectral radius of the injection word]
\label{prop:sr-injection}
The Burau matrix at $t = -1$ of $\sigma_i^2 \sigma_{i+1}^{-1}$ has
characteristic polynomial with roots $\{1, 2+\sqrt{3}, 2-\sqrt{3}\}$.
The dominant eigenvalue is $2 + \sqrt{3} \approx 3.732$.
\end{proposition}

\begin{proof}[Proof sketch]
Working in the unreduced Burau representation at $t = -1$, the
active $2 \times 2$ block of each generator is
\[
  \rho_{-1}(\sigma_i)\big|_{(i,\,i+1)}
  = \begin{pmatrix} 2 & -1 \\ 1 & 0 \end{pmatrix},
\]
with identity on all other strands.  Squaring gives
\[
  \rho_{-1}(\sigma_i^2)\big|_{(i,\,i+1)}
  = \begin{pmatrix} 2 & -1 \\ 1 & 0 \end{pmatrix}^{\!2}
  = \begin{pmatrix} 3 & -2 \\ 2 & -1 \end{pmatrix}.
\]
Since $\det\rho_{-1}(\sigma_j)=1$ for all $j$, the inverse block is
$\rho_{-1}(\sigma_{i+1}^{-1})|_{(i+1,\,i+2)}
 = \bigl(\begin{smallmatrix}0&1\\-1&2\end{smallmatrix}\bigr)$.
Evaluating the product explicitly for $n = 3$ strands yields a
$3 \times 3$ matrix with characteristic polynomial
$(\lambda - 1)(\lambda^2 - 4\lambda + 1) = 0$, giving roots
$1$ and $\lambda = 2 \pm \sqrt{3}$.
\end{proof}

\subsection{Braid word length tradeoffs}
\label{sec:word-tradeoffs}

Beyond shorter length-2 words which slightly reduce complexity, 
a natural alternative is a positive-only length-3 word such as
$\sigma_i^2\sigma_{i+1}$ or $\sigma_i\sigma_{i+1}\sigma_i$.
Both have SR $> 1$ and would formally satisfy the growth
requirement.  However, the mixed-sign choice
$\sigma_i^2\sigma_{i+1}^{-1}$ is load-bearing in two distinct
ways that positive-only words cannot replicate.

\emph{Cross-epoch cancellations.}
The negative exponent $\sigma_{i+1}^{-1}$ creates the possibility
that the trailing generator of one gate firing partially cancels
the leading generator of the next when they occur at adjacent
positions in consecutive epochs.  A positive-only injection word
makes concatenation strictly length-monotone: every firing
lengthens the accumulated word, so word length alone is a
reliable proxy for risk.  Abelian statistics (DWS,
counting-LE, raw firing rate) then track growth faithfully,
and the non-Abelian correction supplied by Burau is marginal.
The mixed-sign word reintroduces cross-epoch cancellations that
Abelian statistics are blind to, making Burau's signal genuinely
irreplaceable.

\emph{Signed interference and spectral resonance.}
The high spectral radius $2 + \sqrt{3}$ arises from the
non-commutative interplay between the positive block $\sigma_i^2$
and the negative block $\sigma_{i+1}^{-1}$.  At $t=-1$, the Burau
matrices for inverse generators introduce negative entries;
Perron--Frobenius no longer applies to the product, and the dominant
eigenvalue need not be real.  This sign interference means that
\emph{rearranging} the same generators changes the spectral radius in
ways that the Abelian gate-count cannot track, exactly the blindness
formalised in \Cref{prop:abelian-blindness}.  For positive-only words
(no $\sigma_i^{-1}$) the matrices are non-negative; Perron--Frobenius
guarantees a real dominant eigenvalue and the Abelian rate is a tighter
proxy.  The discriminator gap that justifies the non-Abelian Burau
pipeline is therefore widest precisely at the mixed-sign gate word
$\sigma_i^2\sigma_{i+1}^{-1}$.

%% ============================================================
\section{Burau Representation and Lyapunov Exponent}
\label{sec:burau}

The construction of \Cref{sec:construction} produces a braid word
that grows longer with each epoch of the random walk. This word is
an algebraic object (a sequence of generators and their
inverses), not a number. To extract a quantitative risk signal, we
need a way to measure how ``tangled'' the braid has become after
$T$ steps.

The Burau representation provides exactly this. It maps each braid
generator to a matrix, converting the braid word into a matrix
product. As the walk progresses and the braid word grows, we
multiply one more matrix into the running product. The key
question becomes: \emph{does this product grow?} If the matrices
merely rotate (eigenvalues on the unit circle), the product stays
bounded and the braid is ``flat''. But
if the product grows exponentially, the braid is genuinely
tangled: each new crossing compounds the complexity of all
previous crossings. The rate of this exponential growth is the
Lyapunov exponent, our primary risk metric.

\paragraph{}
The full Burau computation ($N \times N$ matrices, $N$
simultaneous NHI walkers) is the authoritative discriminator
of the focused/dispersed boundary and the calibration oracle
against which the Abelian rate is evaluated.
In operational deployment, temporal braid LE is the primary
screening instrument; the Burau pipeline provides the ground truth
for confirming the focused/dispersed classification in flagged
ratchets where the Abelian rate may be biased.

\subsection{Burau at $t = -1$: exact integer computation}

\begin{definition}[Unreduced Burau representation]
The \emph{unreduced Burau representation} $\rho_t : B_n \to
\mathrm{GL}_n(\mathbb{Z}[t, t^{-1}])$ maps each generator to:
\[
  \rho_t(\sigma_i) = I_n - t\, E_{ii} + t\, E_{i,i+1}
    + E_{i+1,i} - E_{i+1,i+1},
\]
where $E_{jk}$ is the matrix unit with 1 in position $(j,k)$ and 0
elsewhere.
\end{definition}

At the specialisation $t = -1$, each entry of $\rho_{-1}(\sigma_i)$
is an integer:
\[
  \rho_{-1}(\sigma_i)_{jk} = \begin{cases}
    2  & j = k = i, \\
    -1 & j = i, k = i+1, \\
    1  & j = i+1, k = i, \\
    0  & j = k = i+1, \\
    1  & j = k \notin \{i, i+1\}, \\
    0  & \text{otherwise.}
  \end{cases}
\]
All subsequent matrix products remain integer-valued, enabling exact
computation throughout the walk.

\medskip

\label{rem:integer-overflow}
Working in exact integer arithmetic provides a significant numerical
advantage over a float64 implementation.  The spectral radius of the
accumulated Burau product grows as $(2{+}\sqrt{3})^M$ where $M$ is
the number of gate firings; at $T = 50$ epochs with a typical firing
rate, $M$ can reach several hundred, making $(2{+}\sqrt{3})^M \sim
10^{200}$, which exceeds the float64 maximum ($\approx 1.8\times
10^{308}$) for dense ratchets with high firing rates.  Float64
representations would overflow to \texttt{inf}, yielding no usable
spectral radius.

In Python, the language we use for our calculus, arbitrary-precision integers never overflow:
matrix entries grow large but remain exact, and the characteristic
polynomial is computed exactly over $\mathbb{Z}$.  The only numerical
step is the final eigenvalue extraction from an exact integer matrix,
which is performed in floating point on a matrix whose entries are
already determined exactly.  This design is both the source of the
2.83\% overflow rows reported in \Cref{sec:empirical} (which arise
from fixed-width integer limits in the dataset pipeline, not from
the algorithm itself) and the reason the remaining 97.17\% carry
ground-truth spectral radii rather than float64 approximations.

\subsection{Spectral radius and Lyapunov exponent}

After a walk of length $T$ producing permission braid 
$\beta = \sigma_{i_1}^{\epsilon_1} \cdots \sigma_{i_k}^{\epsilon_k}$,
the \emph{accumulated Burau matrix} is
$\mathbf{B}(T) = \prod_{j=1}^k \rho_{-1}(\sigma_{i_j}^{\epsilon_j})$,
and its \emph{spectral radius} is
$\SR(T) = \max\{|\lambda| : \lambda \text{ eigenvalue of } \mathbf{B}(T)\}$.

\begin{definition}[Lyapunov exponent]
\label{def:le}
The \emph{Lyapunov exponent} of SCC $C$ under WAR assignment
$W$ is
\[
  \LE(C, W) = \frac{\log \SR(T_{\mathrm{high}}) -
    \log \SR(T_{\mathrm{low}})}{T_{\mathrm{high}} - T_{\mathrm{low}}},
\]
estimated over a walk of total length $T_{\mathrm{high}}$ with
$T_{\mathrm{low}} = T_{\mathrm{high}}/2$.
\end{definition}

\begin{definition}[$T$-scaling ratio]
\label{def:tscaling}
The \emph{$T$-scaling ratio} is
\[
  r = \frac{\Delta\log\SR(T_{\mathrm{high}})}
           {\Delta\log\SR(T_{\mathrm{low}})},
\quad
  \Delta\log\SR(T) = \log \SR(T) - \log \SR(T/2).
\]
\end{definition}

Values of $r$ near 1 indicate stable exponential scaling; values
significantly below 1 indicate subexponential or transient growth.

\paragraph{Why exponential growth is stable.}
The NHI walkers form an irreducible finite Markov chain on the SCC, so
the gate-firing process is ergodic.  The sequence
$a_T = \log\|\mathbf{B}(T)\|$ is subadditive by submultiplicativity of
matrix norms, and the shift that advances the walk by one epoch preserves
the stationary measure.  Kingman's subadditive ergodic theorem~\cite{kingman1989} therefore
guarantees that $\frac{1}{T}\log\|\mathbf{B}(T)\|$ converges almost
surely to a deterministic limit $\LE$, independent of the walkers'
starting positions.  Finite-state irreducibility ensures rapid mixing,
explaining why LE stabilises within a few dozen gate firings in practice.

\subsection{Firing-rate lower bound for hub-concentrated topologies}
\label{sec:firing-rate-bound}

The preceding discussion established that LE grows with the
gate-firing rate, but no formal lower bound linked graph topology
to that rate.  This subsection proves such a bound for a natural
class of SCCs: those whose directed edges concentrate on a common
high-privilege hub.

\begin{definition}[$k$-hub SCC]
\label{def:k-hub}
Let $C = (V, E)$ be a ratchet SCC on $|V| = N$ vertices with WAR
assignment $W$.  A vertex $h \in V$ is a \emph{$k$-hub} if there
exist $k \geq 2$ distinct vertices $v_1, \ldots, v_k \in V
\setminus \{h\}$ such that:
\begin{enumerate}
  \item Each $v_j \to h$ is a \emph{directed} edge (i.e.\
        $h \to v_j \notin E$), and
  \item Each such edge is ascending: $W(h) > W(v_j)$ for all
        $j = 1, \ldots, k$.
\end{enumerate}
The \emph{spoke set} is $S = \{v_1, \ldots, v_k\}$ and each
$v_j \to h$ is a \emph{spoke edge}.  Write $d_j$ for the
out-degree of $v_j$ in the full graph $C$.
\end{definition}

\begin{definition}[Spoke-covered $k$-hub SCC]
\label{def:spoke-covered}
A $k$-hub SCC with spoke set $S = \{v_1, \ldots, v_k\}$ is
\emph{spoke-covered} if every vertex $u \in V$ has at least one
directed out-edge to the spoke set: for all $u \in V$,
$\{j : u \to v_j \in E\} \neq \emptyset$.
\end{definition}

Spoke coverage is the condition that the spoke set $S$ is a
\emph{dominating set} of the directed graph: every vertex has at
least one immediate out-neighbour in $S$.  In permission-graph
terms, it says that from any privilege state, some direct
de-escalation path exists, a mild structural condition satisfied
by any ratchet SCC in which the spoke vertices act as
``gateway'' low-privilege roles.

The random walk model (\Cref{sec:construction}) places $n$
walkers performing independent, uniform random walks on~$C$.
At each epoch, every walker at vertex $u$ moves to a neighbour
chosen uniformly at random from the $d_u$ out-neighbours of~$u$.

\begin{lemma}[Spoke-step probability]
\label{lem:spoke-step}
Let $C$ be an SCC containing a $k$-hub $h$ with spoke set
$S = \{v_1, \ldots, v_k\}$, and let $d_{\max}$ denote the
maximum out-degree in $C$.  For $n$
independent walkers, each performing a uniform random walk
on~$C$, let $X_{i,t}$ be the indicator that walker~$i$ takes
a spoke step at epoch~$t$, and define the ensemble spoke-step
frequency
\[
  \hat{p}_h \;=\; \frac{1}{nT}\sum_{i=1}^{n}\sum_{t=1}^{T}
  X_{i,t}.
\]
Then the expected spoke-step probability satisfies:
\begin{enumerate}[label=(\roman*)]
  \item \textbf{Universal bound.}
    \[
      p_h \;=\; \mathbb{E}[\hat{p}_h]
          \;\geq\;
          \frac{k}{|V| \cdot d_{\max}^{\,\mathrm{diam}(C)+1}},
    \]
    where $\mathrm{diam}(C)$ is the diameter of the SCC.
  \item \textbf{Spoke-covered bound.}
    If $C$ is spoke-covered (\Cref{def:spoke-covered}) and
    every vertex has out-edges to all $k$ spoke vertices, then
    \[
      p_h \;\geq\; \frac{k}{d_{\max}^{2}},
    \]
    independent of $|V|$ and $\mathrm{diam}(C)$.
    Under the weaker spoke-coverage condition (each vertex has
    at least one out-edge to $S$), the bound is $p_h \geq
    1/d_{\max}^2$.
\end{enumerate}
The ensemble frequency concentrates around its mean:
\[
  \Pr\!\bigl[\hat{p}_h \leq p_h - \epsilon\bigr]
  \;\leq\; e^{-2n\epsilon^2}
  \quad\text{for all } \epsilon > 0.
\]
\end{lemma}

\begin{proof}See \Cref{app:proofs}.\end{proof}

\begin{lemma}[Spoke-pair event count]
\label{lem:spoke-pair}
Under the hypotheses of \Cref{lem:spoke-step}, let
$m_{\mathrm{spoke}}(t)$ denote the number of unordered pairs
$\{i,j\}$ such that both walker~$i$ and walker~$j$ take a
spoke step at epoch~$t$.  Then
\[
  \mathbb{E}[m_{\mathrm{spoke}}(t)]
  \;\geq\; \binom{n}{2} p_h^2
  \;\geq\; \frac{n(n-1)}{2} \cdot
           \frac{k^2}{N^2 \cdot
           d_{\max}^{\,2(\mathrm{diam}(C)+1)}}.
\]
\end{lemma}

\begin{proof}See \Cref{app:proofs}.\end{proof}

\begin{theorem}[Firing-rate lower bound for $k$-hub SCCs]
\label{thm:firing-rate}
Let $C$ be a ratchet SCC on $N = |V|$ vertices with a $k$-hub
$h$ (\Cref{def:k-hub}), and let $n$ walkers perform independent
uniform random walks on $C$ in stationarity.  Let $p_h$ be the
per-walker spoke-step probability from \Cref{lem:spoke-step}.
Over a walk of $T$ epochs, the expected number of gate-eligible
firings satisfies
\[
  \mathbb{E}[M_T] \;\geq\;
  T \cdot \left\lfloor\frac{n-1}{2}\right\rfloor \cdot p_h^2.
\]
Substituting the two bounds from \Cref{lem:spoke-step}:
\begin{enumerate}[label=(\roman*)]
  \item \textbf{Universal:}
    $\displaystyle\mathbb{E}[M_T] \;\geq\; T \cdot \left\lfloor\frac{n-1}{2}\right\rfloor \cdot
           \frac{k^2}{N^2 \cdot d_{\max}^{\,2(\mathrm{diam}(C)+1)}}.$
  \item \textbf{Spoke-covered:}
    $\displaystyle\mathbb{E}[M_T] \;\geq\; T \cdot \left\lfloor\frac{n-1}{2}\right\rfloor \cdot
           \frac{k^2}{d_{\max}^{4}}.$
\end{enumerate}
For the pipeline's $n = 6$ walkers, $\lfloor(n-1)/2\rfloor = 2$,
so $\mathbb{E}[M_T] \geq 2\,T \cdot p_h^2$ in both cases.
\end{theorem}

\begin{proof}See \Cref{app:proofs}.\end{proof}

\begin{corollary}[Counting-LE lower bound for $k$-hub SCCs]
\label{cor:le-lower-bound}
Under the hypotheses of \Cref{thm:firing-rate}, the
\emph{counting Lyapunov exponent} (the LE that would result
if each gate firing contributed exactly $\log(2 + \sqrt{3})$
to $\log\SR$) satisfies
\[
  \LE_{\mathrm{count}}
  \;\geq\; \frac{\mathbb{E}[M_T]}{T} \cdot
           \log(2 + \sqrt{3})
  \;\geq\; 2\,p_h^2 \cdot \log(2 + \sqrt{3}).
\]
Substituting the two bounds from \Cref{lem:spoke-step} (with $n = N = 6$):
\begin{enumerate}[label=(\roman*)]
  \item \textbf{Universal} (SCC diameter $D = \mathrm{diam}(C)$):
    \[
      \LE_{\mathrm{count}}
      \;\geq\; \frac{k^2 \cdot \log(2 + \sqrt{3})}{
               2 \cdot N^2 \cdot d_{\max}^{\,2(D+1)}}.
    \]
  \item \textbf{Spoke-covered} (every vertex has out-edges to all
    $k$ spokes, independent of $N$ and $D$):
    \[
      \LE_{\mathrm{count}}
      \;\geq\; \frac{k^2 \cdot \log(2 + \sqrt{3})}{
               2 \cdot d_{\max}^{4}}.
    \]
\end{enumerate}
\end{corollary}

\begin{proof}See \Cref{app:proofs}.\end{proof}

\paragraph{Tightness and the temporal correction.}
\label{rem:counting-vs-temporal}
Both bounds in \Cref{cor:le-lower-bound} are on the
\emph{counting} LE, which is an overestimate of the true
(temporal) LE: the shuffling experiment (\Cref{sec:shuffling})
shows temporal ordering produces systematic spectral
cancellation, with temporal $\SR$ below shuffled $\SR$ in 84\%
of pairs.  The bounds therefore provide conservative floors on
$\LE_{\mathrm{temporal}}$ and should be read as order-of-magnitude
indicators rather than tight guarantees.  The median relative gap
between counting and temporal LE is only $1.3\%$
(\Cref{sec:shuffling}), so the bounds remain useful for
identifying clearly high-risk topologies.

\newpage

\section{Braid Non-Commutativity and the Shuffling Experiment}
\label{sec:structure}

A natural question is whether the temporal ordering of gate-firing
injections matters for LE, or whether LE reduces to a simple function
of the gate-firing count. This section shows that the ordering
carries genuine non-commutative signal.

\subsection{Algebraic structure of injection-triplet products}
\label{sec:triplet-algebra}

Let $T_i$ denote the Burau matrix of a single injection triplet
$\sigma_i^2 \sigma_{i+1}^{-1}$ at position~$i$.  Direct computation
at $t=-1$ gives:
\begin{itemize}
  \item \textbf{Same position:} $\SR(T_i \cdot T_i) = (2+\sqrt{3})^2
        \approx 13.93$. Growth is fully multiplicative.
  \item \textbf{Adjacent positions ($|i-j|=1$):}
        $\SR(T_i \cdot T_j) = 1$. The product has all eigenvalues on
        the unit circle (complete spectral cancellation).
  \item \textbf{Gap-2 positions ($|i-j|=2$):}
        $\SR(T_i \cdot T_j) \approx 7.09$.  Growth is partially
        retained but the matrices do not commute:
        $\|T_i T_j - T_j T_i\|_F \approx 7.48$.
\end{itemize}

The adjacent-pair guard (\Cref{sec:gate}) prevents adjacent-position
firings within a single epoch, so within-epoch products are always at
non-adjacent (gap $\geq 2$) triplet positions.  Because the guard
scans positions left-to-right and skips position $i+1$ whenever
position $i$ fires, the within-epoch firing order is fixed by the
scan; there is no ordering degree of freedom to exploit.
Non-commutativity arises \emph{across} epochs, when different
positions fire in consecutive epochs.

\paragraph{Strand count matters.}
With $n$ strands there are $n-2$ generator $\sigma_i^2 \sigma_{i+1}^{-1}$ positions. At $n=4$
(2~positions, always adjacent), every cross-position product gives
$\SR=1$; the guard forces single-position dominance and the permission braid
collapses to a power of one matrix: ordering becomes trivially
irrelevant. At $n=6$ (4~positions), gap-2 pairs $(0,2)$ and $(1,3)$
produce non-commuting products ($\SR \approx 7.09$,
$\|[T_i,T_j]\|_F \approx 7.48$), while the gap-3 pair $(0,3)$
commutes but still has $\SR > 1$.  Together these restore genuine
cross-epoch ordering dependence. Our pipeline uses $n=6$
strands (matching one NHI per SCC vertex).

\subsection{Shuffling experiment}
\label{sec:shuffling}

We ran 469 (SCC, WAR) pairs across all 82 ratchet topologies with 6
NHIs (4 generator positions). For each pair, 30 random walks of
100~epochs were performed. Each walk's braid word was decomposed into
injection triplets, then:
\begin{enumerate}
  \item \textbf{Original:} $\log \SR$ of the temporal-order braid.
  \item \textbf{Shuffled:} triplets uniformly permuted at random
        (20~shuffles per walk, averaged).
  \item \textbf{Position-sorted:} triplets grouped by generator
        position (all position-0 firings first, then position-1,
        etc.), eliminating cross-position temporal cancellation.
\end{enumerate}

\begin{center}
\begin{tabular}{lr}
\toprule
Statistic & Value \\
\midrule
(SCC, WAR) pairs & 469 \\
\addlinespace
\multicolumn{2}{l}{\textit{Original vs.\ Shuffled}} \\
\quad Paired $t$-statistic & $-15.0$ \\
\quad Cohen's $d$ & $-0.69$ \\
\quad Fraction (orig $<$ shuffled) & 84\,\% \\
\quad Median $|$relative difference$|$ & 1.3\,\% \\
\addlinespace
\multicolumn{2}{l}{\textit{Original vs.\ Position-Sorted}} \\
\quad Fraction (orig $<$ sorted) & 97\,\% \\
\quad Mean (orig $-$ sorted) / orig & $-21.5$\,\% \\
\addlinespace
\multicolumn{2}{l}{\textit{Position diversity}} \\
\quad Walks with $\geq 2$ distinct positions & 95.3\,\% \\
\bottomrule
\end{tabular}
\end{center}

The temporal ordering produces a \textbf{significant, systematic}
reduction in SR compared to both shuffled and sorted baselines
($t = -15.0$, $d = -0.69$, $p \ll 10^{-30}$).
Temporal braids have lower SR in 84\% of pairs, and 22\% lower
than position-sorted braids on average.

\subsection{Interpretation: spectral cancellation as signal}

The direction of the effect (temporal ordering \emph{reduces}
SR) has a clean explanation. In the temporal braid, a walker that
fires at position~$i$ in epoch~$t$ is likely to fire at a nearby
position in epoch~$t+1$ (since strands move locally).  These
cross-epoch adjacent-position interactions produce partial spectral
cancellation ($\SR(T_i \cdot T_{i\pm 1}) = 1$).  Shuffling
randomises these interactions, breaking the temporal correlation
and eliminating the systematic cancellation.

This effect is not noise: it carries structural information about
the SCC. Topologies where directed edges cluster in adjacent
regions of the braid produce more cross-epoch cancellation and
therefore lower LE at the same firing rate. The temporal braid
captures this spatial structure; the shuffled braid destroys it.

\subsection{Temporal LE is a better discriminator than gate-firing count}
\label{sec:discriminator}

A natural objection is that if the firing rate explains most of the
variance in LE, the full Burau validation machinery is overkill: one could simply
count gate firings and multiply by $\log(2{+}\sqrt{3})$.  We refute
this with three tests.

\paragraph{Test 1: Partial correlation with DWS.}
$\DWS$ is computed from topology and WAR alone, independent of the
Burau validation computation.  We compute $r(\LE, \DWS \mid n_{\text{trips}})$,
the Pearson correlation between LE and DWS after linearly removing
the effect of the gate-firing count.

\begin{center}
\begin{tabular}{lrl}
\toprule
LE variant & $r(\LE, \DWS \mid n_{\text{trips}})$ & $p$ \\
\midrule
$\LE_{\text{temporal}}$  & $\mathbf{0.175}$  & $1.5 \times 10^{-4}$ \\
$\LE_{\text{counting}}$  & $0.001$            & $0.98$ \\
$\LE_{\text{shuffled}}$  & $0.149$            & $1.2 \times 10^{-3}$ \\
\bottomrule
\end{tabular}
\end{center}

After removing firing-rate variance, \textbf{counting LE retains zero
correlation with DWS} ($r = 0.001$, $p = 0.98$).  Temporal LE
retains $r = 0.175$ ($p < 10^{-3}$).  The non-commutative correction
carries structural information that counting cannot capture.

\paragraph{Test 2: Matched-rate spread.}
Grouping pairs into firing-rate bins (±5~triplets) and measuring the
within-bin LE range shows that temporal LE consistently exhibits
\textbf{3--4$\times$ more spread} than counting LE at the same
firing rate.  For example, at $\approx 90$~triplets per walk
($n = 40$ pairs): temporal range = 0.534, counting range = 0.127.
This additional spread reflects structural differences between SCCs
that share a firing rate but differ in the spatial distribution of
their directed edges.

\paragraph{Test 3: DWS-sign Fisher ratio.}
Using DWS sign (positive vs.\ non-positive) as an external
discriminant axis independent of the Burau computation:
\begin{center}
\begin{tabular}{lr}
\toprule
LE variant & Fisher ratio (DWS$>$0 vs.\ DWS$\leq$0) \\
\midrule
$\LE_{\text{temporal}}$  & \textbf{0.088} \\
$\LE_{\text{shuffled}}$  & 0.083 \\
$\LE_{\text{counting}}$  & 0.061 \\
\bottomrule
\end{tabular}
\end{center}
Temporal LE separates DWS-sign groups 45\% better than counting LE.

%% ============================================================
\section{Two-Regime Classification}
\label{sec:classification}

\subsection{LE as the sole regime discriminator}

Each ratchet deployment, a specific $(C, W)$ pair, is assigned to one
of two dynamical regimes by a single quantity: the Lyapunov exponent
$\LE$ of the accumulated Burau matrix product.

\medskip
\noindent
\textbf{$\LE$ (primary signal)} captures \emph{structural risk}: how
fast does the permission braid spectral radius grow per step?  It is determined by
the graph topology (density of directed edges, hub connectivity) and by
the WAR assignment (which edges are traversed in ascending order).  Its
value encodes whether non-commutative cancellations between consecutive
gate firings dominate (low $\LE$, focused regime) or whether crossings
accumulate without cancellation (high $\LE$, dispersed regime).

\begin{definition}[Two-regime classification]
\label{def:regimes}
Each ratchet deployment $(C, W)$ belongs to exactly one of two regimes,
separated at the corpus median $\LE = 0.5847$ (the median of per-SCC
mean LE across the 1000-SCC corpus, chosen so that the two regimes are
equally represented under a random-SCC assumption):
\[
\begin{array}{r|c}
  & \LE \\
\hline
\LE < 0.5847 & \textbf{focused} \\
\LE \geq 0.5847 & \textbf{dispersed} \\
\end{array}
\]
\end{definition}

The reason why we are choosing the median is that given a random uniform sample of SCCs, 
there is no bias towards either topology.

\begin{table}[h]
\centering
\begin{tabular}{lll}
\toprule
Regime & Topology & Remediation \\
\midrule
\textbf{focused} &
  Escalation channelled, few paths &
  WAR reassignment \\
\textbf{dispersed} &
  Hub-rich, multiple independent paths &
  Topology restructuring \\
\bottomrule
\end{tabular}
\caption{Two-regime classification with remediation paths.}
\label{tab:regimes}
\end{table}

The \textbf{dispersed} regime ranks highest because it persists across
WAR assignments: hub topologies sustain the gate regardless of privilege
alignment, so no WAR reassignment suppresses escalation.  Only adding
or reversing directed edges (topology surgery) suffices.

The \textbf{focused} regime is WAR-controlled: the current assignment
aligns directed edges with ascending privilege flow, driving spectral
growth.  Reassigning WAR to break that alignment reduces $\LE$ and
remediates the ratchet without touching the graph structure.

\medskip

\label{rem:dws}
The Directed WAR Sum $\DWS(C, W) = \sum_{(u \to v) \in E_{\mathrm{dir}}}
(W(v) - W(u))$ is a static, $O(|E_{\mathrm{dir}}|)$ abelian statistic
that measures net uphill privilege flow under the current assignment.
$\DWS > 0$ indicates that directed edges run net-uphill, amplifying
gate activity; $\DWS \leq 0$ indicates suppression.  Within the focused
regime, $\DWS$ helps identify the specific WAR reassignment that
remediates the deployment.  However, $\DWS$ is an abelian count and
cannot determine $\LE$: two deployments with the same $\DWS$ and the
same gate-firing rate can have different $\LE$ values and therefore
belong to different regimes.  $\LE$ must be determined by a non-abelian
instrument (\Cref{rem:abelian-blindness}).

\begin{proposition}[Abelian blindness]
\label{prop:abelian-blindness}
No function of the abelian statistics of a ratchet walk (generator
frequencies, gate-firing rate, or $\DWS$) determines the Lyapunov
exponent $\LE$.
\end{proposition}

\begin{proof}
The Lyapunov exponent $\LE = \lim_{T\to\infty}\frac{1}{T}
\mathbb{E}[\log\|\rho(g_T)\cdots\rho(g_1)\|]$ is a property of the
\emph{joint} distribution of the random matrix sequence, not of its
one-dimensional marginals (the generator frequencies).  Any abelian
statistic is a function of those marginals alone.  Since the Burau
matrices for adjacent generators do not commute, rearranging the
\emph{order} of generators, while preserving their multiset, changes
the spectral growth rate.

\emph{Explicit witness.}  In $B_3$ at the specialisation $t=-1$, set
$A = \rho(\sigma_1) = \bigl(\begin{smallmatrix}1&1\\0&1\end{smallmatrix}\bigr)$
and
$C = \rho(\sigma_2^{-1}) = \bigl(\begin{smallmatrix}1&0\\1&1\end{smallmatrix}\bigr)$.
Both are integer upper/lower-triangular matrices; their product is not the identity
($AC \neq CA$).  Consider the two words of length~5:
\begin{align*}
  w_1 &= \sigma_1^3\,\sigma_2^{-2}, &
  \rho(w_1) &= A^3C^2 = \begin{pmatrix}7&3\\2&1\end{pmatrix},
  \quad \SR(\rho(w_1)) = 4+\sqrt{15};\\
  w_2 &= (\sigma_1\sigma_2^{-1})^2\sigma_1, &
  \rho(w_2) &= (AC)^2A = \begin{pmatrix}5&8\\3&5\end{pmatrix},
  \quad \SR(\rho(w_2)) = 5+2\sqrt{6}.
\end{align*}
Both words have identical abelian image $(3,-2)\in\mathbb{Z}^2$
(three $\sigma_1$'s and two $\sigma_2^{-1}$'s), so every abelian
statistic (gate-firing rate, $\DWS$, exponent sums) agrees.  Yet
$4+\sqrt{15}\approx 7.87 \neq 9.90\approx 5+2\sqrt{6}$, so the
periodic walks repeating $w_1$ and $w_2$ have Lyapunov exponents
$\tfrac{1}{5}\log(4+\sqrt{15})\approx 0.413$ and
$\tfrac{1}{5}\log(5+2\sqrt{6})\approx 0.459$ per step, respectively.
No abelian statistic can distinguish them.
\end{proof}

\begin{corollary}[Telescoping blindness for directed cycles]
\label{cor:telescoping}
Let $v_1 \to v_2 \to \cdots \to v_k \to v_1$ be a directed cycle embedded in a
ratchet SCC $(C, W)$.  The contribution of this cycle to $\DWS$ is identically
zero for \emph{every} WAR assignment $W$:
\[
  \DWS_{\mathrm{cycle}}(W)
  \;=\; \sum_{i=1}^{k}\bigl(W(v_{i+1\bmod k}) - W(v_i)\bigr)
  \;=\; 0.
\]
No WAR assignment can make a directed cycle net-ascending
($\DWS_{\mathrm{cycle}} > 0$) or net-descending
($\DWS_{\mathrm{cycle}} < 0$): the cycle is $\DWS$-neutral by
construction.

Although the cycle's net privilege flow is always zero, individual edges within
it can be ascending or descending depending on the WAR assignment.  Whenever
at least one cycle edge is ascending, the gate fires on that edge, injecting
non-commuting braid generators that drive spectral growth.  A directed cycle
with $r$ ascending edges fires $r$ gate events per traversal while
contributing exactly zero to $\DWS$, so the abelian pre-filter is
permanently blind to the cycle's escalation potential regardless of WAR.
\end{corollary}

\begin{proof}
The sum telescopes: $(W(v_2)-W(v_1)) + (W(v_3)-W(v_2)) + \cdots +
(W(v_1)-W(v_k)) = W(v_1) - W(v_1) = 0$.  Since the WAR values appear
and cancel in adjacent pairs around the cycle, the result holds for
every choice of $W$.
\end{proof}

\label{rem:abelian-blindness}
Abelian statistics ($\DWS$, gate-firing rate) are blind to the
\emph{order} in which directed edges are traversed and therefore cannot
detect the non-commutative cancellations that define the
focused/dispersed boundary.  Two deployments with identical $\DWS$ and
firing rate can have different $\LE$ values: one focused (cancellations
dominate), one dispersed (crossings accumulate), with opposite
remediations.

\Cref{cor:telescoping} gives the sharpest structural instance of this
blindness.  A directed cycle contributes exactly zero to $\DWS$ for
\emph{every} WAR assignment simultaneously; it is not that a specific
unlucky assignment happens to cancel: the cancellation is
topological and unavoidable.  Any SCC whose directed edges include a
cycle will therefore have a $\DWS$-invisible escalation channel that
the permission braid can detect but the abelian pre-filter never will.

\paragraph{Regime as a deployment property.}
The two regimes classify a \emph{deployment}, not a topology alone.
Two SCCs with identical primorial invariant $R$ (same fusion class)
may sit in different regimes if their WAR assignments differ.
Conversely, the same SCC observed at different epochs may transition
between regimes as its WAR assignment drifts.  This is precisely the
operational value of the braid framework: topology screening (primorial)
identifies ratchets, but deployment screening (Burau LE) identifies
which ratchets are actively dangerous and which kind of intervention
is required.

\section{Empirical Validation}
\label{sec:empirical}

\subsection{Dataset}

The validation dataset~\cite{kaggle} covers 1,000 ratchet SCC topologies
generated at random from the set of all 6-vertex SCC topologies,
each a directed--mixed graph combining directed and bidirectional edges.
For each topology, 50 WAR assignments $(W_0,\ldots,W_5)\in\{0,\ldots,5\}^6$
are sampled uniformly at random and evaluated with the braid walk
simulator at $T = 50$ epochs and $w = 64$ walks, yielding \textbf{49,972}
(SCC, WAR) pairs in total.

Two independent regime labels are computed per pair:
\begin{itemize}
  \item \textbf{Burau LE}: dispersed if $\LE > 0.5847$, focused otherwise.
        The threshold is the median of the per-SCC mean $\LE$ across the
        full 1,000-SCC corpus.
  \item \textbf{Abelian rate}: dispersed if $\texttt{le\_counting} > 0.6586$,
        focused otherwise, with the threshold set identically as the
        corpus median of per-SCC mean rates.
\end{itemize}

Across all 49,972 pairs, \textbf{2,852 (5.7\%)} exhibit a regime
disagreement between the two metrics.  The two error directions are
nearly symmetric: 1,418 \emph{FD} cases (Focused by Burau, Dispersed
by rate; Burau says the deployment is manageable, rate raises a false
alarm) and 1,434 \emph{DF} cases (Dispersed by Burau, Focused by
rate; Burau detects genuine non-commutative growth that the abelian
rate misses).  This near-symmetry at the corpus level conceals
a sharply topology-dependent pattern: individual SCCs are strongly
skewed toward one error type, and the skew direction is governed by
whether the topology's directed edges concentrate crossings on a
single braid strand position or fan out to produce diverse non-commuting
generators.

At $T = 50$ epochs, 1,412 rows (2.83\%) encounter integer arithmetic
overflow in the Burau matrix product; these rows record
$\LE = \texttt{le\_counting} = 0$ and therefore agree trivially on the
focused label, contributing zero disagreements.  No SCC has all 50
WAR samples overflowed, and all four case-study SCCs below are
overflow-free.

\subsection{Case studies}
The four subsections that follow use ratchet topologies drawn directly
from the synthetic corpus. They are ordered by failure mode:
\textbf{SCC~216} (9 dir, 2 bidir) shows pure FD failures: directed edges
concentrate ascending crossings on a small set of strand positions, and
the abelian rate consistently over-calls dispersed.
\textbf{SCC~207} (10 dir, 4 bidir) shows pure DF failures: a bidir hub
at $v_1$ distributes crossing events across diverse strand positions,
driving Burau growth that the abelian rate misses.
\textbf{SCC~126} (9 dir, 3 bidir) sits at the transition: depending on WAR
assignment, the topology produces either FD or DF failures with equal frequency.
\textbf{SCC~233} (6 dir, 5 bidir) shows pure DF failures, driven by a directed 4-cycle
whose WAR contributions telescope to zero while its braid generators
grow exponentially.

\subsection{SCC~216: Directed Fan-Out Hub, FD-Only}
\label{sec:scc216}

SCC~216 has 9 directed and 2 bidirectional edges.
The directed edges are
$v_1{\to}v_0$, $v_0{\to}v_2$, $v_0{\to}v_3$, $v_1{\to}v_2$,
$v_3{\to}v_1$, $v_5{\to}v_1$, $v_3{\to}v_4$, $v_3{\to}v_5$, $v_4{\to}v_5$;
the bidirectional edges are $v_2{\leftrightarrow}v_4$ and $v_2{\leftrightarrow}v_5$.
The dominant structural feature is $v_3$ as a directed fan-out hub:
it sends edges to $v_1$, $v_4$, and $v_5$, while itself receiving only
from $v_0{\to}v_3$.  Ascending crossings from $v_3$ repeatedly land on
$v_3$'s strand, and the two bidir edges ($v_2{\leftrightarrow}v_4$,
$v_2{\leftrightarrow}v_5$) oscillate locally without diversifying generator
strand positions.

\begin{figure}[h]
\centering
\begin{tikzpicture}[->,>=Stealth,shorten >=1pt,auto,semithick]
  \node[state,fill=lightgray] (v0) at (0,   3)   {$v_0$};
  \node[state,fill=lightgray] (v1) at (2.6, 1.5) {$v_1$};
  \node[state,fill=lightgray] (v2) at (2.6,-1.5) {$v_2$};
  \node[state,fill=lightgray] (v3) at (0,  -3)   {$v_3$};
  \node[state,fill=lightgray] (v4) at (-2.6,-1.5){$v_4$};
  \node[state,fill=lightgray] (v5) at (-2.6, 1.5){$v_5$};
  \path
    %% directed edges
    (v1) edge (v0)
    (v0) edge (v2)
    (v0) edge (v3)
    (v1) edge (v2)
    (v3) edge (v1)
    (v5) edge (v1)
    (v3) edge (v4)
    (v3) edge (v5)
    (v4) edge (v5)
    %% bidirectional edges
    (v2) edge [bend left=10,opacity=0.4] (v4)
    (v4) edge [bend left=10,opacity=0.4] (v2)
    (v2) edge [bend left=10,opacity=0.4] (v5)
    (v5) edge [bend left=10,opacity=0.4] (v2)
  ;
\end{tikzpicture}
\caption{SCC~216: 9 directed (solid) and 2 bidirectional edges (paired, 40\% opacity).
  $v_3$ is the directed fan-out hub; $v_2{\leftrightarrow}v_4$ and
  $v_2{\leftrightarrow}v_5$ are local bidir oscillators.}
\label{fig:scc216}
\end{figure}
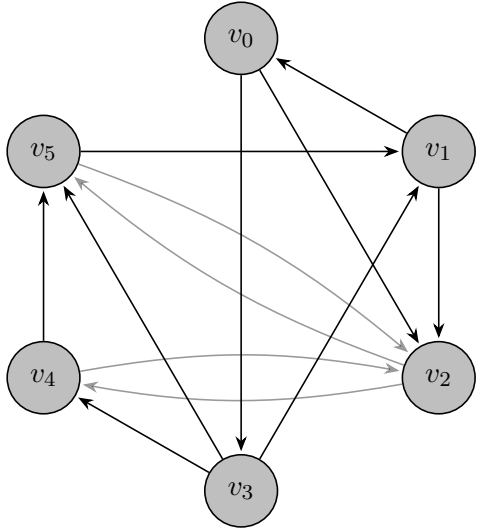

Out of 50 WAR samples, \textbf{6 are FD} and \textbf{0 are DF}
(thresholds $\theta_{\LE}=0.5847$, $\theta_{\mathrm{rate}}=0.6586$).
The abelian rate consistently over-calls dispersed; Burau never over-calls.

\begin{table}[h]
\centering
\begin{tabular}{lrrl}
\toprule
WAR $(W_0,\ldots,W_5)$ & $\LE$ & rate & type \\
\midrule
$(0,2,0,0,3,4)$ & 0.460 & 0.675 & FD \\
$(0,3,0,0,4,5)$ & 0.471 & 0.689 & FD \\
$(4,2,1,2,4,2)$ & 0.532 & 0.858 & FD \\
$(5,0,0,5,4,0)$ & 0.551 & 0.793 & FD \\
$(4,0,0,3,5,0)$ & 0.570 & 0.910 & FD \\
$(1,1,0,5,2,5)$ & 0.509 & 0.921 & FD \\
\bottomrule
\end{tabular}
\caption{FD disagreements for SCC~216 (no DF failures).
  FD: $\LE < 0.5847$ (focused by Burau) and rate $> 0.6586$ (dispersed by abelian counter).}
\label{tab:scc216}
\end{table}

\paragraph{Why only FD?}
In every FD case WAR elevates $v_3$ or its targets ($v_4$, $v_5$) to high WAR,
causing ascending crossings from $v_3{\to}v_1$, $v_3{\to}v_4$, $v_3{\to}v_5$
to fire the gate repeatedly.
The abelian counter tallies each gate-fire and reports high rate.
The Burau product, however, accumulates generators at a concentrated set of
strand positions: $v_3$ is the structural origin of nearly every directed
crossing, so all injected braid generators commute or nearly commute.
Repeated applications of generators at the same strand position are not
exponentially amplifying; spectral radius stays subcritical ($\LE < 0.5847$
across all 6 FD cases, reaching 0.570 at most).
DF failure never occurs because the bidir connections of $v_2$ ($v_2{\leftrightarrow}v_4$,
$v_2{\leftrightarrow}v_5$) do not create diverse independent crossing chains:
no WAR assignment can route flow through enough distinct strand positions to
drive exponential non-commutative growth.

\subsection{SCC~207: Bidir Hub Diversifies Generators, DF-Only}
\label{sec:scc207}

SCC~207 has 10 directed and 4 bidirectional edges.
The directed edges are
$v_0{\to}v_2$, $v_3{\to}v_0$, $v_4{\to}v_0$, $v_0{\to}v_5$,
$v_3{\to}v_1$, $v_2{\to}v_3$, $v_4{\to}v_2$, $v_2{\to}v_5$,
$v_3{\to}v_5$, $v_4{\to}v_5$;
the bidirectional edges are
$v_0{\leftrightarrow}v_1$, $v_1{\leftrightarrow}v_2$,
$v_1{\leftrightarrow}v_5$, $v_3{\leftrightarrow}v_4$.
Node $v_1$ is the bidir hub, connected bidirectionally to $v_0$, $v_2$, and
$v_5$; nodes $v_3$ and $v_4$ share a bidir link while each sending multiple
directed edges outward.

\begin{figure}[h]
\centering
\begin{tikzpicture}[->,>=Stealth,shorten >=1pt,auto,semithick]
  \node[state,fill=lightgray] (v0) at (0,   3)   {$v_0$};
  \node[state,fill=lightgray] (v1) at (2.6, 1.5) {$v_1$};
  \node[state,fill=lightgray] (v2) at (2.6,-1.5) {$v_2$};
  \node[state,fill=lightgray] (v3) at (0,  -3)   {$v_3$};
  \node[state,fill=lightgray] (v4) at (-2.6,-1.5){$v_4$};
  \node[state,fill=lightgray] (v5) at (-2.6, 1.5){$v_5$};
  \path
    %% directed edges
    (v0) edge (v2)
    (v3) edge (v0)
    (v4) edge (v0)
    (v0) edge (v5)
    (v3) edge (v1)
    (v2) edge (v3)
    (v4) edge (v2)
    (v2) edge (v5)
    (v3) edge (v5)
    (v4) edge (v5)
    %% bidirectional edges
    (v0) edge [bend left=10,opacity=0.4] (v1)
    (v1) edge [bend left=10,opacity=0.4] (v0)
    (v1) edge [bend left=10,opacity=0.4] (v2)
    (v2) edge [bend left=10,opacity=0.4] (v1)
    (v1) edge [bend left=10,opacity=0.4] (v5)
    (v5) edge [bend left=10,opacity=0.4] (v1)
    (v3) edge [bend left=12,opacity=0.4] (v4)
    (v4) edge [bend left=12,opacity=0.4] (v3)
  ;
\end{tikzpicture}
\caption{SCC~207: 10 directed (solid) and 4 bidirectional edges (paired, 40\% opacity).
  $v_1$ is the bidir hub connected to $v_0$, $v_2$, $v_5$;
  $v_3$ and $v_4$ are directed fan-out sources sharing a bidir link.}
\label{fig:scc207}
\end{figure}
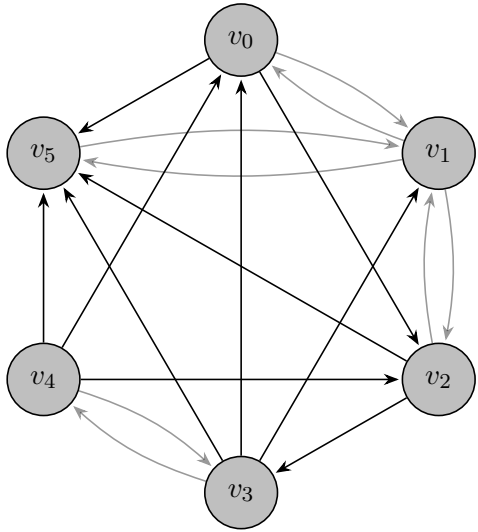

Out of 50 WAR samples, \textbf{0 are FD} and \textbf{5 are DF}.
The abelian rate never over-calls dispersed; Burau consistently detects
non-commutative growth that the rate misses.

\begin{table}[h]
\centering
\begin{tabular}{lrrl}
\toprule
WAR $(W_0,\ldots,W_5)$ & $\LE$ & rate & type \\
\midrule
$(1,3,5,0,5,2)$ & 0.585 & 0.618 & DF \\
$(4,5,1,4,5,3)$ & 0.596 & 0.635 & DF \\
$(0,0,5,5,2,4)$ & 0.618 & 0.638 & DF \\
$(0,0,5,3,2,2)$ & 0.591 & 0.611 & DF \\
$(0,3,3,2,5,2)$ & 0.630 & 0.653 & DF \\
\bottomrule
\end{tabular}
\caption{DF disagreements for SCC~207 (no FD failures).
  DF: $\LE > 0.5847$ (dispersed by Burau) and rate $< 0.6586$ (focused by abelian counter).}
\label{tab:scc207}
\end{table}

\paragraph{Why only DF?}
In each DF case WAR distributes so that $v_3$ or $v_4$ (directed fan-out
sources) are elevated relative to $v_0$, $v_2$, $v_5$.
The directed edges $v_3{\to}v_0$, $v_4{\to}v_0$, $v_3{\to}v_1$,
$v_4{\to}v_2$, $v_3{\to}v_5$, $v_4{\to}v_5$ fan out to five different target
nodes, each occupying a distinct strand position in the braid.
The bidir hub $v_1$ adds further crossing diversity by oscillating between
$v_0$, $v_2$, and $v_5$ at different WAR levels.
The combined effect is a braid word with generators at many independent
positions; non-commuting products accumulate and Burau detects spectral
growth ($\LE > 0.5847$).
The abelian counter, however, sees near-balanced directed flow in these
cases (rates marginally below threshold), classifying as focused while
the non-commutative structure grows exponentially.
FD failure never occurs because the fan-out structure always distributes
crossing events across enough strand positions to prevent generator repetition
from dominating the Burau product.

\subsection{SCC~126: Dual-Source Hub, Balanced FD and DF}
\label{sec:scc126}

SCC~126 has 9 directed and 3 bidirectional edges.
The directed edges are
$v_1{\to}v_0$, $v_0{\to}v_3$, $v_5{\to}v_0$, $v_1{\to}v_2$,
$v_1{\to}v_4$, $v_5{\to}v_1$, $v_3{\to}v_2$, $v_5{\to}v_2$, $v_3{\to}v_5$;
the bidirectional edges are
$v_2{\leftrightarrow}v_4$, $v_3{\leftrightarrow}v_4$, $v_4{\leftrightarrow}v_5$.
The topology has two major directed sources: $v_1$ (sending to $v_0$, $v_2$, $v_4$)
and $v_5$ (sending to $v_0$, $v_1$, $v_2$).  Node $v_4$ is the bidir hub,
connected bidirectionally to $v_2$, $v_3$, and $v_5$.

\begin{figure}[h]
\centering
\begin{tikzpicture}[->,>=Stealth,shorten >=1pt,auto,semithick]
  \node[state,fill=lightgray] (v0) at (0,   3)   {$v_0$};
  \node[state,fill=lightgray] (v1) at (2.6, 1.5) {$v_1$};
  \node[state,fill=lightgray] (v2) at (2.6,-1.5) {$v_2$};
  \node[state,fill=lightgray] (v3) at (0,  -3)   {$v_3$};
  \node[state,fill=lightgray] (v4) at (-2.6,-1.5){$v_4$};
  \node[state,fill=lightgray] (v5) at (-2.6, 1.5){$v_5$};
  \path
    %% directed edges
    (v1) edge (v0)
    (v0) edge (v3)
    (v5) edge (v0)
    (v1) edge (v2)
    (v1) edge (v4)
    (v5) edge (v1)
    (v3) edge (v2)
    (v5) edge (v2)
    (v3) edge (v5)
    %% bidirectional edges
    (v2) edge [bend left=10,opacity=0.4] (v4)
    (v4) edge [bend left=10,opacity=0.4] (v2)
    (v3) edge [bend left=12,opacity=0.4] (v4)
    (v4) edge [bend left=12,opacity=0.4] (v3)
    (v4) edge [bend left=10,opacity=0.4] (v5)
    (v5) edge [bend left=10,opacity=0.4] (v4)
  ;
\end{tikzpicture}
\caption{SCC~126: 9 directed (solid) and 3 bidirectional edges (paired, 40\% opacity).
  $v_1$ and $v_5$ are dual directed sources; $v_4$ is the bidir hub.}
\label{fig:scc126}
\end{figure}
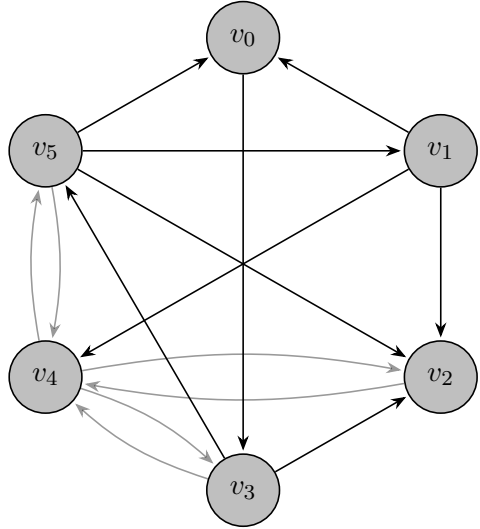

Out of 50 WAR samples, \textbf{3 are FD} and \textbf{3 are DF},
the most balanced split of the four case studies.

\begin{table}[h]
\centering
\begin{tabular}{lrrl}
\toprule
WAR $(W_0,\ldots,W_5)$ & $\LE$ & rate & type \\
\midrule
$(5,1,0,2,0,2)$ & 0.468 & 0.681 & FD \\
$(4,3,2,5,2,3)$ & 0.428 & 0.746 & FD \\
$(3,2,4,3,3,3)$ & 0.546 & 0.785 & FD \\
\midrule
$(1,1,3,5,5,5)$ & 0.605 & 0.648 & DF \\
$(2,3,3,4,0,3)$ & 0.638 & 0.651 & DF \\
$(0,2,1,5,2,5)$ & 0.604 & 0.614 & DF \\
\bottomrule
\end{tabular}
\caption{Disagreements for SCC~126.
  FD: $\LE < 0.5847$ and rate $> 0.6586$; DF: $\LE > 0.5847$ and rate $< 0.6586$.}
\label{tab:scc126}
\end{table}

\paragraph{FD failures: $v_1$ as a crossing funnel.}
In the FD cases (e.g.\ $(4,3,2,5,2,3)$, $\LE{=}0.428$, rate $=0.746$)
WAR elevates $v_3$ to the highest rank while $v_1$ has intermediate WAR.
Directed edges $v_5{\to}v_1$, $v_3{\to}v_2$, and $v_3{\to}v_5$ produce
ascending crossings that land repeatedly near $v_3$'s strand position.
The abelian counter accumulates many ascending events and reports high rate.
The Burau product, however, sees generators concentrated at the strand
corresponding to $v_3$; repeated identical generators do not drive spectral
growth ($\LE{=}0.428$).

\paragraph{DF failures: $v_5$ and $v_4$ activate diverse chains.}
In the DF cases (e.g.\ $(1,1,3,5,5,5)$, $\LE{=}0.605$, rate $=0.648$)
WAR elevates both $v_3$, $v_4$, and $v_5$ to high WAR.
Now $v_5{\to}v_0$, $v_5{\to}v_1$, $v_5{\to}v_2$, and the bidir
$v_4{\leftrightarrow}v_5$ crossings span five distinct strand positions.
The diverse generators produce non-commuting products; Burau detects
spectral growth ($\LE{=}0.605$) while the abelian rate, seeing near-balanced
directed flow, falls just below the dispersed threshold.

\paragraph{Hub as regime switch.}
The balanced FD/DF split arises from the interplay of $v_1$ and $v_5$: when
WAR concentrates flow through $v_1$, it becomes a crossing funnel~$\Rightarrow$~FD;
when WAR distributes to activate $v_5$'s fan-out plus $v_4$'s bidir connections,
generator diversity drives exponential Burau growth~$\Rightarrow$~DF.
The abelian rate is blind to this distinction in both directions.

\subsection{SCC~233: Directed 4-Cycle Within Bidir-Heavy Topology, DF-Only}
\label{sec:scc233}

SCC~233 has 6 directed and 5 bidirectional edges: it exhibits
\textbf{0 FD} and \textbf{7 DF} failures out of 50 WAR samples.

The directed edges are $v_0{\to}v_1$, $v_4{\to}v_0$, $v_5{\to}v_0$,
$v_3{\to}v_1$, $v_1{\to}v_5$, $v_5{\to}v_4$;
the bidirectional edges are
$v_0{\leftrightarrow}v_3$, $v_1{\leftrightarrow}v_2$, $v_1{\leftrightarrow}v_4$,
$v_2{\leftrightarrow}v_4$, $v_3{\leftrightarrow}v_5$.
The defining structural feature is the directed 4-cycle
$v_0{\to}v_1{\to}v_5{\to}v_4{\to}v_0$, embedded within a bidir-rich neighbourhood.

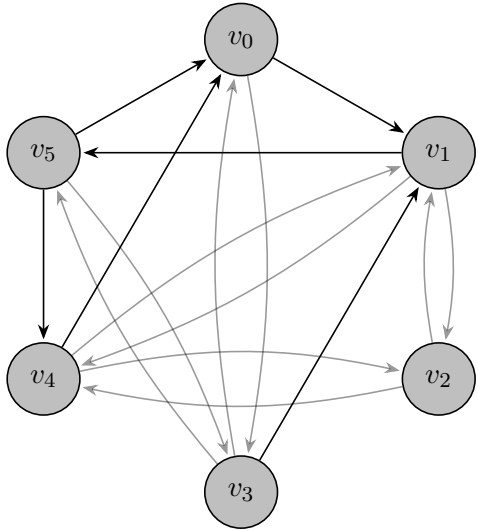
\begin{figure}[h]
\centering
\begin{tikzpicture}[->,>=Stealth,shorten >=1pt,auto,semithick]
  \node[state,fill=lightgray] (v0) at (0,   3)   {$v_0$};
  \node[state,fill=lightgray] (v1) at (2.6, 1.5) {$v_1$};
  \node[state,fill=lightgray] (v2) at (2.6,-1.5) {$v_2$};
  \node[state,fill=lightgray] (v3) at (0,  -3)   {$v_3$};
  \node[state,fill=lightgray] (v4) at (-2.6,-1.5){$v_4$};
  \node[state,fill=lightgray] (v5) at (-2.6, 1.5){$v_5$};
  \path
    %% directed edges (4-cycle: v0->v1->v5->v4->v0)
    (v0) edge (v1)
    (v1) edge (v5)
    (v5) edge (v4)
    (v4) edge (v0)
    %% remaining directed edges
    (v5) edge (v0)
    (v3) edge (v1)
    %% bidirectional edges
    (v0) edge [bend left=10,opacity=0.4] (v3)
    (v3) edge [bend left=10,opacity=0.4] (v0)
    (v1) edge [bend left=10,opacity=0.4] (v2)
    (v2) edge [bend left=10,opacity=0.4] (v1)
    (v1) edge [bend left=10,opacity=0.4] (v4)
    (v4) edge [bend left=10,opacity=0.4] (v1)
    (v2) edge [bend left=12,opacity=0.4] (v4)
    (v4) edge [bend left=12,opacity=0.4] (v2)
    (v3) edge [bend left=10,opacity=0.4] (v5)
    (v5) edge [bend left=10,opacity=0.4] (v3)
  ;
\end{tikzpicture}
\caption{SCC~233: 6 directed (solid) and 5 bidirectional edges (paired, 40\% opacity).
  Directed 4-cycle $v_0{\to}v_1{\to}v_5{\to}v_4{\to}v_0$ is the source of DF failures;
  bidir edges $v_1{\leftrightarrow}v_2$, $v_1{\leftrightarrow}v_4$, $v_2{\leftrightarrow}v_4$
  form a local triangle.}
\label{fig:scc233}
\end{figure}

\begin{table}[h]
\centering
\begin{tabular}{lrrl}
\toprule
WAR $(W_0,\ldots,W_5)$ & $\LE$ & rate & type \\
\midrule
$(1,4,5,2,3,3)$ & 0.587 & 0.600 & DF \\
$(2,3,1,3,5,1)$ & 0.601 & 0.625 & DF \\
$(0,4,2,1,4,2)$ & 0.606 & 0.635 & DF \\
$(4,4,4,4,3,1)$ & 0.627 & 0.657 & DF \\
$(2,3,5,5,2,1)$ & 0.603 & 0.656 & DF \\
\bottomrule
\end{tabular}
\caption{DF disagreements for SCC~233 (no FD failures; 2 of 7 rows omitted for brevity).
  DF: $\LE > 0.5847$ (dispersed by Burau) and rate $< 0.6586$ (focused by abelian counter).}
\label{tab:scc233}
\end{table}

\paragraph{Why only DF?}
\Cref{cor:telescoping} provides the structural explanation.  The four
directed cycle edges $v_0{\to}v_1{\to}v_5{\to}v_4{\to}v_0$ form a
closed loop, so their WAR differences telescope to exactly zero:
\[
  \DWS_{\mathrm{cycle}}
  = \bigl(W(v_1){-}W(v_0)\bigr)
  + \bigl(W(v_5){-}W(v_1)\bigr)
  + \bigl(W(v_4){-}W(v_5)\bigr)
  + \bigl(W(v_0){-}W(v_4)\bigr)
  = 0
\]
for \emph{every} WAR assignment.  This is a topological identity that holds
unconditionally, not a coincidental numerical cancellation.
The abelian rate therefore sees the cycle's contribution as permanently neutral;
it can only report dispersed based on the two non-cycle directed edges
($v_5{\to}v_0$ and $v_3{\to}v_1$), which in the DF cases nearly cancel
or fall below threshold.

Yet the cycle still fires the gate on its ascending edges.  In
$(1,4,5,2,3,3)$: $v_0{=}1$, $v_1{=}4$, so $v_0{\to}v_1$ is ascending;
$v_1{=}4$, $v_5{=}3$, so $v_1{\to}v_5$ is descending; $v_5{=}3$, $v_4{=}3$,
so $v_5{\to}v_4$ is flat; $v_4{=}3$, $v_0{=}1$, so $v_4{\to}v_0$ is
descending.  The ascending edge $v_0{\to}v_1$ fires the gate, injecting
braid generators.  As walkers traverse the cycle, the four structurally
distinct generator positions produce non-commuting products that accumulate
exponentially ($\LE{=}0.587$).
The bidir triangle $v_1{\leftrightarrow}v_2{\leftrightarrow}v_4{\leftrightarrow}v_1$
adds further generator diversity at $v_1$'s strand without inflating the
abelian rate.

FD failure never occurs because the telescoping identity prevents the cycle
from contributing to the abelian rate: any WAR assignment that fires the
cycle must produce a rate contribution near zero from cycle activity alone,
so the abelian counter cannot over-call dispersed based on cycle crossings.
The bidir-heavy neighbourhood could in principle produce FD by inflating the
abelian rate, but bidir edges' paired forward--reverse crossings cancel
in $\DWS$ as well, providing no net directed-flow signal to over-count.

%% ============================================================
\section{Discussion}
\label{sec:discussion}

\subsection{Topology classification vs.\ deployment classification}

The primorial invariant classifies \emph{topologies}: it discriminates
oscillators (no ratchet structure, topologically safe) from ratchets
(directed WAR paths capable of monotone escalation).  Two SCCs in the
same primorial fusion class share identical cycle escalation profiles
and receive one analyst review, the right level of abstraction for
triage and whitelist management at scale.

The braid framework classifies \emph{deployments}.  It is applied
\emph{after} primorial screening, for SCCs already identified as
ratchets.  Its input is the topology together with the ground-truth
WAR assignment; its output is $\LE$, the non-abelian risk metric that
abelian statistics cannot determine.  $\LE$ alone partitions every
ratchet deployment into focused or dispersed, the two regimes with
distinct remediations.  The two layers operate sequentially
(see \Cref{sec:operability} for the full pipeline).
This separation avoids redundant dynamic analysis: oscillators are
filtered before any dynamic analysis runs.

\subsection{Alternative Abelian baselines}
\label{sec:abelian-baselines}
The Abelian baseline used throughout this paper, gate-firing count
and $|\DWS|$, is intentionally simple: it represents the strongest
signal available from a single linear scan of the directed edges under
a given WAR assignment.  A natural question is whether more
sophisticated Abelian statistics could close the gap with temporal
braid $\LE$.  We examine three candidates and show that none can
substitute for the non-Abelian screen, for reasons that are structural
(\Cref{prop:abelian-blindness}), empirical, or both.

\paragraph{Spectral gap of the transition matrix.}
The spectral gap $\gamma = 1 - \lambda_2(P)$ of the random-walk
transition matrix $P$ measures mixing speed: a large gap means fast
mixing, so walkers move broadly across the SCC and diversify generator
positions; a small gap means slow mixing, concentrating the walk in
local neighbourhoods and repeating similar generators.  This creates a
correlation pathway with the focused/dispersed distinction (fast-mixing
SCCs lean dispersed, slow-mixing lean focused), but the correlation is
indirect and operates through the same gate-firing mechanism already
captured by the firing count.

More importantly, $\gamma$ is a purely topological quantity computed
from $P$ alone, independent of the WAR assignment $W$.  It therefore
cannot detect the WAR-dependent component of risk that $\DWS$
captures, and combining $\gamma$ with $\DWS$ yields a two-variable
Abelian composite more complex than the current baseline without
guaranteed improvement.  Fundamentally, $\gamma$ is an eigenvalue of a
commutative matrix operation and is therefore an Abelian statistic in
the sense of \Cref{prop:abelian-blindness}: it provably cannot
determine $\LE$.

\paragraph{Cycle-count heuristics.}
It is tempting to count directed cycles~\cite{johnson1975} (or simple directed cycles of bounded length)
in the SCC as a topological structure proxy: more
cycles imply more opportunities for generator diversity and higher
$\LE$.  However, cycle counts are purely topological and entirely blind
to the WAR assignment.  They are therefore \emph{strictly more naive}
than $|\DWS|$, which integrates both topology (which edges are
directed) and scalar (WAR values along those edges).  Exact cycle
enumeration is \#P-hard in general; approximations add computational
cost while providing less information than the Abelian baseline already
in use.  We do not consider cycle-count heuristics a competitive
baseline.

\paragraph{WAR-weighted directed centrality.}
A more solid abelian alternative targets the
\emph{crossing-concentration} mechanism directly.  For each node $v$,
define its \emph{ascending in-load}
\[
  \ell(v, W)
  \;=\;
  \sum_{\substack{(u \to v)\,\in\, E_{\mathrm{dir}} \\[2pt] W(v)\,>\,W(u)}}
  \bigl(W(v) - W(u)\bigr),
\]
so that $\DWS(C,W) = \sum_{v} \ell(v,W)$.  The \emph{WAR-weighted
centrality index} $H$ measures how unevenly that load is distributed:
\[
  H(C, W)
  \;=\;
  \mathrm{Gini}\!\left(\bigl\{\ell(v, W)\bigr\}_{v \in V}\right)
  \;=\;
  \frac{
    \displaystyle\sum_{v,\,u \in V}
    \bigl|\ell(v,W) - \ell(u,W)\bigr|
  }{
    2\,|V|\,\DWS(C,W)
  },
\]
with $H=0$ for uniform load and $H=1$ for a single-node hub.
High $H$ predicts focused behaviour; $H$ is undefined when $\DWS \le 0$.
 
We tested $H$ on all FD and DF disagreement samples from
SCCs~216 and~207, classifying as focused when $H > 0.5$:

\begin{table}[h]
\begin{center}
\begin{tabular}{lcc}
\toprule
Error type & rate correct & $H$ correct \\
\midrule
FD: Burau focused, rate false alarm (6 cases, SCC~216) & 0/6 & \textbf{4/6} \\
DF: Burau dispersed, rate misses it (5 cases, SCC~207) & 0/5  & 1/5 \\
\midrule
Total & 0/11 & 5/11 \\
\bottomrule
\end{tabular}
\caption{The centrality index: a local, specialized discriminator}
\end{center}
\end{table}

$H$ improves on the abelian rate for FD cases in SCC~216: ascending
load concentrates on $v_3$'s strand, and the cases where $\DWS > 0$
(cases with $\DWS \in \{4,5\}$) are correctly classified as focused
by $H > 0.5$.  The four FD cases with $\DWS < 0$ have undefined $H$
and fall back to the default focused label.

However, this success is \emph{structurally narrow}: $H$ is valid only when (i)~$\DWS > 0$ (so
that $\ell$ is well-defined and positive) and (ii)~the topology is
hub-shaped, routing ascending flow through a single dominant node.
SCC~216 satisfies both conditions in the DWS$>$0 subset; the
broader 50\,000-sample corpus shows that DWS~$\le~0$ in a majority
of FD cases, making $H$ inapplicable for most of them.  

On DF cases $H$ fails almost entirely (1/5):
dispersed behaviour arises from generator diversity along the
\emph{paths feeding} the hub, a path-level ordering property that any
node-aggregate discards by construction.

\subsection{SCCs with the same primorial R}

As discussed in~\cite{primorial2026}, two SCCs in the same fusion class
share the same R vector: same cycle count, same ordinal WAR
relationships along each cycle edge.  Within such a class, $\LE$ can
still differ between instances: instances sharing the same ordinal
structure may have different cardinal WAR magnitudes, placing them in
different regimes.  This is the precise sense in which the braid
framework adds information that the primorial invariant cannot provide.

\subsection{Limitations}

\begin{itemize}
\item \textbf{Abelian rate is a one-sided proxy with asymmetric errors.}
      The gate-firing rate cannot detect non-commutative cancellations
      and is an \emph{upper bound} on $\LE$ (cancellation factor
      $c \le 1$ across all 49\,972 rows).  This creates a qualitative
      asymmetry between the two error types: FD errors (rate false
      alarms, $n = 1\,418$) carry mean $c = 0.725$ and large
      confidence gaps, meaning the rate would confidently prescribe
      unnecessary topology surgery on a WAR-remediable deployment.
      DF errors (rate misses, $n = 1\,434$) are always marginal
      ($c = 0.957$) and operationally lower-cost.  \textbf{When $\LE$
      and the rate disagree with conviction, $\LE$ is always the
      correct call.}  Temporal braid $\LE$ retains partial correlation
      $r = 0.175$ with deployment structure ($p < 10^{-3}$), versus
      $r = 0.001$ ($p = 0.98$) for counting LE.

\item \textbf{The non-Abelian advantage is primarily a false-alarm
      suppressor.}
      The shuffling experiment (\Cref{sec:shuffling}) shows a median
      $\log\SR$ gap of only $1.3\%$; the non-commutative correction
      explains $3\%$ of residual variance over the full corpus.  The
      benefit concentrates on the FD side: cancellation drops to
      $c \approx 0.73$ in hub-concentrated topologies, where the rate
      would confidently misfire, versus $c \approx 0.96$ on the DF
      side.  

\item \textbf{Synthetic validation sample.}
      The 1,000-SCC corpus is not extracted from production deployments. Burau
      $\LE$ as a focused/dispersed discriminator should be treated as
      a working hypothesis, not a validated security primitive.

\item \textbf{$B_5$ Burau representation is not faithful.}
      For $n \ge 5$~\cite{bigelow1999}, distinct braids may share a
      Burau matrix, making $\LE$ a lower bound on true braid
      complexity.  The operational impact is limited: kernel elements
      are long structured words unlikely to arise from short
      randomized injections, non-faithfulness can only suppress
      spectral growth (downward bias only), and the firing-rate and
      induced-LE bounds bypass the unfaithful kernel entirely.
\end{itemize}

\newpage
%% ============================================================
\section{Operability}
\label{sec:operability}  
  
The braid framework operates on SCCs already identified as  
\emph{ratchets} by the primorial invariant; topological oscillators  
are filtered out before any dynamic analysis runs.  Within ratchets,  
the two regimes are not merely classification labels: each implies a  
distinct mental model and a concrete intervention target.  The  
critical operational boundary is between \textbf{focused} and  
\textbf{dispersed}: both require action but the action type is  
opposite.  
 
\paragraph{} 
Burau LE defines this boundary.

\subsection{Regime-specific mental models}

\paragraph{Focused: the channelled ratchet.}
Escalation paths exist but traffic is concentrated through few
directed edges.  Non-commutative cancellations in the Burau product
moderate spectral growth, keeping $\LE$ below the dispersed threshold.
The WAR assignment plays a key role: focused ratchets are WAR-sensitive,
meaning that reassigning privileges to break the ascending alignment of
directed edges will reduce $\LE$ and suppress the gate.  The mental
model is a funnel that can be blocked by removing the funnel's gradient.

\paragraph{Dispersed: the irreducible tangle.}
Hub-rich topology provides multiple independent escalation paths.
Crossings accumulate without cancelling in the braid product, driving
sustained high spectral growth.  No WAR reassignment can suppress the
gate because the hub topology sustains concurrent ascending traversals
regardless of privilege alignment.  The risk is in the graph structure
itself.  The mental model is an irreducible tangle: its complexity
persists under any labelling.

\subsection{Intervention summary}

\begin{table}[h]
\centering
\begin{tabular}{llll}
\toprule
Regime & Mental model & Intervention target & Operator \\
\midrule
\textbf{focused}    & channelled ratchet  & WAR assignment    & reassign WAR \\
\textbf{dispersed}  & irreducible tangle  & edge set          & add/remove edges \\
\bottomrule
\end{tabular}
\caption{Operability guide: mental model and intervention target per
regime.  Both regimes apply to ratchets (primorial-screened).
The focused/dispersed boundary is defined by Burau LE (impossibility
separator, calibration oracle).}
\label{tab:operability}
\end{table}

The intervention cost escalates monotonically: focused ratchets require
an IAM-level privilege reassignment; dispersed ratchets require
architectural-level changes to the permission graph, which may require
review beyond standard IAM operations.

\subsection{Complementarity with the primorial invariant}

The primorial invariant and the braid framework answer different
questions at different levels of the analysis pipeline:

\begin{itemize}
  \item \textbf{Primorial (topology level): oscillators vs.\ ratchets.}
        The primorial invariant encodes the cycle structure of an SCC
        as a vector of rational numbers.  Its output is a binary
        verdict: \emph{oscillator} (topologically safe, no further
        analysis) or \emph{ratchet} (escalation risk, proceed).
        Ratchets with identical escalation profiles are grouped into
        \emph{fusion classes}, compressing the triage space.

  \item \textbf{Braid framework (deployment level): focused vs.\ dispersed.}
        Given a ratchet SCC and its ground-truth WAR assignment, the
        braid framework outputs $\LE$, which determines whether the
        deployment is focused (WAR-remediable) or dispersed
        (topology-remediable).  This determination is provably beyond
        Abelian statistics; only a non-Abelian instrument can supply it.
\end{itemize}

The two layers compose sequentially into a three-tier pipeline:
{\small
\[
  \text{Primorial}
  \xrightarrow{\text{\scriptsize ratchet?}}
  \text{Ratchet SCC}
  \xrightarrow{\substack{\text{\scriptsize rate (fast triage)}\\[-2pt]\text{\scriptsize temporal LE (non-Abelian screen)}}}
  \text{focused / dispersed}
  \to
  \text{Remediation.}
\]
}
Burau LE sits \emph{off the operational path}: it is not run in real-time screening.

Neither layer alone is sufficient.  Primorial cannot assess deployed
severity: the same ratchet topology can span both regimes depending on
WAR assignment.  Braids cannot replace the topological screen: running
braid analysis on every SCC without primorial's oscillator/ratchet
filter would waste computation on structures incapable of escalation.
Together they form a complete, two-layer risk pipeline with Burau LE
as the theoretical anchor.

\section{Future work}

Three directions extend the theoretical and practical reach of the
framework.

\paragraph{Faithful representations.}
The Burau representation at $t = -1$ is unfaithful for $n \geq 5$
(\Cref{sec:discussion}), which limits the validation
to a lower bound on braid complexity.  The Lawrence--Krammer--Bigelow
(LKB) representation~\cite{krammer2002} is faithful
for all~$n$
Replacing the Burau pipeline with LKB would eliminate the kernel
issue entirely, at the cost of working with
$\binom{n}{2} \times \binom{n}{2}$ matrices ($15 \times 15$ at
$n = 6$).
A natural first step is to validate that the LE values
produced by the LKB representation agree with the Burau values for
the 82 ratchet topologies; any discrepancy would quantify the
information lost to unfaithfulness and determine whether the
$\approx 6\times$ matrix-size increase is justified operationally.

\paragraph{Extending the firing-rate lower bound.}
The bound of \Cref{thm:firing-rate} assumes a single high-privilege hub.
Two natural generalisations remain open.
\emph{(a)~Multi-hub topologies:} sum contributions of disjoint hub
substructures separated by at least two rank positions (avoiding
cross-cluster cancellation).
\emph{(b)~Directed-cycle topologies:} \Cref{cor:telescoping} establishes
that directed cycles have $\DWS_{\mathrm{cycle}} = 0$ for all WAR
assignments, making them a class where Abelian screening is provably
useless.  The open problem is not the blindness (which is now
proven) but a tighter \emph{firing-rate lower bound}: SCCs whose
directed edges form a cycle $v_1\to\cdots\to v_k\to v_1$ produce
deterministic gate-firing patterns, and the resulting bounds should
be tighter than the hub model, since walkers traverse ascending edges
in a predictable rotation rather than via a random spoke-step.
Both extensions could narrow the three-order-of-magnitude gap between the
universal lower bound and observed LE in dense ratchets.

\paragraph{Injection word optimisation.}
The current framework fixes the injection word $\sigma_i^2\sigma_{i+1}^{-1}$,
selected for its spectral radius $2+\sqrt{3}$ and its mixed-sign structure.
The mixed sign is load-bearing: the negative exponent creates cross-epoch
cancellations that Abelian statistics cannot detect, and this is precisely
the mechanism behind the non-Abelian advantage ($r=0.175$ partial
correlation vs.\ $r=0.001$ for counting LE).

One could ask whether sweeping the exponent pair $(a, b)$ in
$\sigma_i^a \sigma_{i+1}^b$ with $a > 0$, $b < 0$ could increase this
marginal gain.  Two effects pull in opposite directions.
\emph{Higher $|b|$} deepens cross-epoch cancellations, potentially
widening the gap between temporal and shuffled SR and amplifying the
non-Abelian correction.
\emph{Higher $|a|$} raises the per-firing SR contribution, increasing
the raw LE signal but not necessarily the fraction of that signal
invisible to Abelian counting.

%% ============================================================
\section{Conclusion}\label{sec:conclusion}  
  
No Abelian statistic can locate the boundary between focused and  
dispersed ratchets: Abelian measures are blind to edge-interleaving  
order and therefore cannot detect the non-commutative cancellations  
that define the boundary (\Cref{prop:abelian-blindness}).  This is  
not a quantitative gap to be closed by better counting; it is  
structural.  Burau LE is the separator that proves the gap and the  
calibration oracle that sets the threshold.  

\paragraph{}  
Averaged over 49\,972 (SCC, WAR) pairs from a 1\,000-SCC corpus,  
the non-commutative correction is modest: temporal braid LE retains  
partial correlation $r = 0.175$ with deployment structure after  
controlling for firing rate ($p < 10^{-3}$), while counting LE  
collapses to $r = 0.001$ ($p = 0.98$).  But the $5.7\%$ of pairs  
where Burau and the Abelian rate disagree are not uniformly  
scattered: they concentrate at the focused/dispersed decision  
boundary with systematic, topology-dependent error direction. 
At this boundary, the Abelian rate carries no residual signal and  
would systematically prescribe the wrong remediation. The
non-Abelian screen is redundant for deployments deep in either  
regime but irreplaceable for the ambiguous cases that determine  
intervention type.  
  
%% ============================================================
\bibliographystyle{plain}

\newpage

\appendix

\section{Proofs for lower-bound growth rate}
\label{app:proofs}

\begin{proof}[Proof of \Cref{lem:spoke-step}]

Fix a walker~$i$ and an epoch $t \in \{1, \ldots, T\}$.

\paragraph{Step 1: Spoke-step probability in terms of spoke-set occupation.}
Condition on the position of walker~$i$ at epoch~$t$.  If the
walker is at vertex $v_j \in S$, it has at least one
outgoing spoke edge ($v_j \to h$) and at most $d_{\max}$
outgoing edges total, so the probability of choosing the spoke
edge is at least $1/d_{\max}$.  If the walker is at any vertex
outside~$S$, the probability of a spoke step is zero.  Hence,
\[
  \mathbb{E}[X_{i,t}]
  \;\geq\;
  \frac{1}{d_{\max}}
  \sum_{j=1}^{k}
  \Pr(\text{walker } i \text{ is at } v_j
       \text{ at epoch } t).
\]
Summing over all walkers and epochs and dividing by $nT$:
\[
  p_h
  \;\geq\;
  \frac{1}{d_{\max}}
  \sum_{j=1}^{k}
  \underbrace{\frac{1}{nT}
  \sum_{i=1}^{n}\sum_{t=1}^{T}
  \Pr(\text{walker } i \text{ is at } v_j
       \text{ at epoch } t)}_{\displaystyle =:\;\bar\pi(v_j)}.
\]
Since $C$ is strongly connected, the walk is irreducible with unique
stationary distribution~$\pi$.  In stationarity $\bar\pi(v_j) = \pi(v_j)$;
for arbitrary initialisation, $\bar\pi(v_j) \to \pi(v_j)$ as $T\to\infty$
by the ergodic theorem.  It therefore suffices to lower-bound
$\sum_j \pi(v_j)$, the total stationary mass on the spoke set.

\paragraph{Step 2a: Universal bound on $\pi(v_j)$ via path-chaining.}
Since $\sum_v \pi(v) = 1$, some vertex $v^+$ satisfies
$\pi(v^+) \geq 1/|V|$.  Because $C$ is strongly connected,
$v^+$ reaches any spoke vertex $v_j$ via a directed path
$x_0 = v^+ \to x_1 \to \cdots \to x_\ell = v_j$ of length
$\ell \leq \mathrm{diam}(C)$.  At each step, the stationary
equation $\pi(w) = \sum_{u \to w} \pi(u)/d^+(u)$ implies
$\pi(x_{m+1}) \geq \pi(x_m)/d_{\max}$ (the contribution from
$x_m$ alone, all other terms being non-negative).  Iterating:
\[
  \pi(v_j) \;\geq\; \frac{1}{|V| \cdot d_{\max}^{\mathrm{diam}(C)}}.
\]
Substituting into Step~1:
\[
  p_h \;\geq\;
  \frac{k}{|V| \cdot d_{\max}^{\,\mathrm{diam}(C)+1}}.
  \tag{i}
\]

\paragraph{Step 2b: Spoke-coverage bound via summed stationarity.}
The path-chaining argument bounds each $\pi(v_j)$ separately and
discards all inflow to $v_j$ except from one predecessor.
When $C$ is spoke-covered, we can instead bound
$\sum_j \pi(v_j)$ \emph{globally} by summing the stationary
equations over the entire spoke set:
\[
  \sum_{j=1}^{k} \pi(v_j)
  \;=\; \sum_{j=1}^{k} \sum_{u:\, u \to v_j} \frac{\pi(u)}{d^+(u)}
  \;=\; \sum_{u \in V} \pi(u)
        \cdot \frac{|\{j : u \to v_j \in E\}|}{d^+(u)}.
\]
If every vertex has out-edges to all $k$ spoke vertices
($|\{j : u \to v_j\}| = k$ for all $u$) and $d^+(u) \leq d_{\max}$:
\[
  \sum_{j=1}^{k} \pi(v_j)
  \;\geq\; k \sum_{u \in V} \frac{\pi(u)}{d_{\max}}
  \;=\; \frac{k}{d_{\max}}.
\]
Under the weaker spoke-coverage condition
($|\{j : u \to v_j\}| \geq 1$ for all $u$), the right-hand
side is $1/d_{\max}$.  Substituting the full-coverage case
into Step~1:
\[
  p_h \;\geq\; \frac{1}{d_{\max}} \cdot \frac{k}{d_{\max}}
       \;=\; \frac{k}{d_{\max}^{2}}.
  \tag{ii}
\]
This bound is independent of $|V|$ and $\mathrm{diam}(C)$:
it avoids the path-chaining penalty entirely by using the
full inflow to the spoke set from \emph{every} vertex
simultaneously.

\paragraph{Step 3: Concentration across walkers.}
Within a single walker, the indicators $X_{i,1}, \ldots,
X_{i,T}$ are Markov-dependent (the walker's position at
epoch~$t+1$ depends on its position at epoch~$t$), so
Hoeffding's inequality cannot be applied directly across epochs.
However, the $n$ walkers are \emph{independent}.

Define the per-walker spoke-step frequency
$Y_i = \frac{1}{T}\sum_{t=1}^{T} X_{i,t}$ for each walker
$i = 1, \ldots, n$.  Each $Y_i \in [0, 1]$, and the
$Y_1, \ldots, Y_n$ are independent (since the walkers perform
independent random walks).  The ensemble average is
$\hat{p}_h = \frac{1}{n}\sum_{i=1}^{n} Y_i$.  Applying
Hoeffding's inequality to the $n$ independent bounded random
variables $Y_1, \ldots, Y_n$:
\[
  \Pr\!\bigl[\hat{p}_h \leq \mathbb{E}[\hat{p}_h] -
  \epsilon\bigr]
  \;\leq\; \exp\!\bigl(-2n\epsilon^2\bigr).
\]
The firing-rate theorem (\Cref{thm:firing-rate}) uses the
\emph{expected} firing count over $T$ epochs: by linearity of
expectation, $\mathbb{E}[M_T]$ is exact regardless of
within-walker temporal dependence, so the lower bound requires
no concentration assumption.  The Hoeffding bound above
provides an additional guarantee: as the number of walkers~$n$
grows, the empirical ensemble rate $\hat{p}_h$ concentrates
around its expectation, ensuring robustness against
pathological individual walks.  At the framework's $n = 6$,
the concentration is modest; at larger $n$ the bound tightens
as $\exp(-\Omega(n))$.
\end{proof}

\begin{proof}[Proof of \Cref{lem:spoke-pair}]

Let $A_i(t)$ be the indicator that walker~$i$ takes a spoke
step at epoch~$t$.  Since the walkers perform independent
random walks, $\Pr[A_i(t)\cap A_j(t)] = \Pr[A_i(t)]\cdot
\Pr[A_j(t)] \geq p_h^2$ for every pair $i\neq j$.  By
linearity of expectation:
\[
  \mathbb{E}[m_{\mathrm{spoke}}(t)]
  \;=\;\sum_{\{i,j\}\subseteq[n]}\Pr[A_i(t)\cap A_j(t)]
  \;\geq\;\binom{n}{2}p_h^2.
\]
Substituting the bound from \Cref{lem:spoke-step} gives the
second inequality.
\end{proof}

\begin{proof}[Proof of \Cref{thm:firing-rate}]

Label walkers by NHI identifier $1 < 2 < \cdots < n$ and fix
the $\lfloor(n-1)/2\rfloor$ non-overlapping pairs
$\mathcal{P} = \bigl\{\{1,2\},\{3,4\},\ldots\bigr\}$.

\emph{Each pair fires with probability $\geq p_h^2$.}
For $\{i,i+1\}\in\mathcal{P}$, the walkers are independent so
$\Pr[A_i(t)\cap A_{i+1}(t)]\geq p_h^2$.  When both fire, both
traverse ascending spoke edges and both gate conditions
(\Cref{def:gate}) are satisfied.

\emph{Hub-landing guarantees rank-adjacency.}
Every spoke edge terminates at $h$, so both walkers arrive at
the same vertex with the same WAR value $W(h)$.  Ties in WAR
are broken by NHI identifier, placing the two walkers in
consecutive ranks: the gate fires for this pair.

\emph{Non-overlapping pairs do not conflict.}
When multiple pairs in $\mathcal{P}$ fire simultaneously,
all hub-landing walkers are ordered by NHI index.  The pairs
$\{1,2\},\{3,4\},\ldots$ occupy non-overlapping rank positions,
so the adjacent-pair guard (\Cref{sec:gate}) cannot block any
pair in $\mathcal{P}$ due to a firing by another.

\emph{Conclusion.}
The $\lfloor(n-1)/2\rfloor$ firing events are independent
(disjoint subsets of independent walkers) and each contributes
probability $\geq p_h^2$.  By linearity of expectation, summed
over $T$ epochs:
\[
  \mathbb{E}[M_T] \;\geq\;
  T\cdot\left\lfloor\frac{n-1}{2}\right\rfloor\cdot p_h^2.
  \qedhere
\]
\end{proof}

\begin{proof}[Proof of \Cref{cor:le-lower-bound}]

Each gate firing multiplies the Burau product by the injection
word's matrix, whose dominant eigenvalue is $2 + \sqrt{3}$
(\Cref{prop:sr-injection}).  In the counting-LE model (which
ignores cross-position spectral cancellation), $M_T$ firings
contribute $M_T \cdot \log(2 + \sqrt{3})$ to $\log\SR(T)$.
Dividing by $T$ gives the per-epoch rate.  The numerical bound
follows from \Cref{thm:firing-rate} with
$p_h \geq k/(N \cdot d_{\max}^{D+1})$ from
\Cref{lem:spoke-step}.
\end{proof}

\end{document}